\documentclass[11pt, letterpaper]{article}
\usepackage{graphicx} 
\usepackage[
letterpaper,
top=1in,
bottom=1in,
left=1in,
right=1in]{geometry}

\usepackage[utf8]{inputenc}

\usepackage{amsmath, amsthm, amssymb}
\usepackage{mathtools}

\usepackage[dvipsnames]{xcolor}
\usepackage[T1]{fontenc}
\usepackage[full]{textcomp}
\usepackage[american]{babel}

\usepackage{caption}
\usepackage{booktabs} 
\usepackage{siunitx}  

\usepackage{newpxtext} 
\usepackage{textcomp} 
\usepackage[varg,bigdelims]{newpxmath}
\usepackage[scr=rsfso]{mathalfa}
\usepackage{bm}

\usepackage[pagebackref,colorlinks=true,urlcolor=blue,linkcolor=blue,citecolor=OliveGreen]{hyperref}
\usepackage[nameinlink]{cleveref}
\crefname{equation}{}{}

\usepackage{dsfont}

\usepackage{enumitem}

\usepackage{multirow}

\usepackage{algorithm}
\usepackage{algpseudocode}
\algnewcommand{\Input}[1]{\State \textbf{Input:} #1}
\algnewcommand{\Output}[1]{\State \textbf{Output:} #1}

\usepackage{graphicx}

\usepackage{thmtools}
\usepackage{thm-restate}

\newtheorem{theorem}{Theorem}[section]
\newtheorem{corollary}[theorem]{Corollary}
\newtheorem{lemma}[theorem]{Lemma}
\newtheorem{claim}[theorem]{Claim}

\newtheorem{proposition}[theorem]{Proposition}

\newtheorem{fact}[theorem]{Fact}

\theoremstyle{definition}
\newtheorem{definition}[theorem]{Definition}

\newtheorem{example}[theorem]{Example}

\newtheorem{remark}[theorem]{Remark}

\newcommand{\poly}{\mathrm{poly}}

\renewcommand{\epsilon}{\varepsilon}

\newcommand{\R}{\mathbb{R}}
\newcommand{\N}{\mathbb{N}}
\newcommand{\Q}{\mathbb{Q}}
\newcommand{\Z}{\mathbb{Z}}

\newcommand{\mU}{\mathcal{U}}
\newcommand{\mV}{\mathcal{V}}
\newcommand{\mW}{\mathcal{W}}
\newcommand{\mL}{\mathcal{L}}
\newcommand{\mZ}{\mathcal{Z}}

\newcommand{\x}{x}
\newcommand{\y}{y}

\newcommand{\bc}{\mathrm{bl}}
\newcommand{\rank}{\mathrm{rank}}
\newcommand{\enc}[1]{\mathbf{enc}({#1})}

\newcommand{\mat}[1]{[#1]}

\newcommand{\mynabla}[1]{\nabla_{\! {#1}}\,}
\newcommand{\spann}[1]{\mathrm{span}\{#1\}}
\newcommand{\supp}[1]{\mathrm{supp}(#1)}

\newcommand{\recc}{\mathrm{recc}}
\newcommand{\vol}{\mathrm{vol}}
\newcommand{\aff}{\mathrm{aff}}
\newcommand{\conv}{\mathrm{conv}}

\hypersetup{
    pdfauthor = {Lucas Slot, David Steurer, Manuel Wiedmer},
}
\title{Hesse's Redemption: \\ Efficient Convex Polynomial Programming}

\author{
  \begin{minipage}[t]{0.3\textwidth}
    \centering
    Lucas Slot \\ \footnotesize \href{mailto:lucas.slot@inf.ethz.ch}{lucas.slot@inf.ethz.ch} \\ \small ETH Zurich
  \end{minipage}%
  \hfill
  \begin{minipage}[t]{0.3\textwidth}
    \centering
    David Steurer \\ \footnotesize \href{mailto:david.steurer@inf.ethz.ch}{david.steurer@inf.ethz.ch} \\ \small ETH Zurich
  \end{minipage}%
  \hfill
  \begin{minipage}[t]{0.3\textwidth}
    \centering
    Manuel Wiedmer \\ \footnotesize \href{mailto:manuel.wiedmer@inf.ethz.ch}{manuel.wiedmer@inf.ethz.ch} \\ \small ETH Zurich
  \end{minipage}
}

\date{November 5, 2025}

\begin{document}

\maketitle
\thispagestyle{empty}

\begin{abstract}
Efficient algorithms for convex optimization, such as the ellipsoid method, require an \emph{a priori} bound on the radius of a ball around the origin guaranteed to contain an optimal solution if one exists.
For linear and convex quadratic programming, such solution bounds follow from classical characterizations of optimal solutions by systems of \emph{linear} equations.
For other programs, e.g., semidefinite ones, examples due to Khachiyan show that optimal solutions may require huge coefficients with an exponential number of bits, even if we allow approximations. Correspondingly, semidefinite programming is not even known to be in \textbf{NP}.

The unconstrained minimization of convex polynomials of degree four and higher has remained a fundamental open problem between these two extremes: its optimal solutions do not admit a linear characterization and, at the same time, Khachiyan-type examples do not apply.
We resolve this problem by developing new techniques to prove solution bounds when no linear characterizations are available.
Even for programs minimizing a convex polynomial (of arbitrary degree) over a polyhedron, we prove that the existence of an optimal solution implies that an approximately optimal one with polynomial bit length also exists.
These solution bounds, combined with the ellipsoid method, yield the first polynomial-time algorithm for convex polynomial programming, settling a question posed by Nesterov \mbox{(Math. Program., 2019)}. Before, no polynomial-time algorithm was known even for unconstrained minimization of a convex polynomial of degree four.

Our results rely on a structural decomposition of any convex polynomial into a sum of a linear function and a polynomial on a linear subspace that admits a strongly convex lower bound,
where the logarithm of the strong convexity parameter is polynomially bounded in the input size.
A key component of our proof is a strong local-to-global property for convex polynomials:
if at every point \emph{some} directional second derivative vanishes, then a \emph{single} directional second derivative must vanish everywhere.
While Hesse erroneously claimed that this property holds for general polynomials (J. Reine Angew. Math., 1851), we show that it holds for convex ones.

\end{abstract}

\setcounter{tocdepth}{2}
\newpage
\thispagestyle{empty}
\tableofcontents
\newpage

\setcounter{page}{1}

\section{Introduction}
A central result in convex programming is the algorithmic equivalence of optimization and separation,
emerging from a body of work around the ellipsoid method~\cite{YudinNemirovski1976, Khachiyan:LP, Schrijver:Ellipsoid, Karp_1980, padberg:inria-00076483, GrötschelLovaszSchrijver:ellipsoid}.
However, this equivalence crucially requires \emph{effective solution bounds}, even when we allow approximations:
if optimal solutions exist, then (approximately) optimal solutions must exist such that the logarithm of their norm is polynomially bounded in the input size.
The ellipsoid method requires such solution bounds when determining the radius of its initial ellipsoid and the running time depends polynomially on the logarithm of this radius.\footnote{This radius is sometimes called ``big $R$''. 
  The ellipsoid method has another parameter, often referred to as ``little~$r$'', relating to the existence of a small ball inside the feasible region.
  However, the need for ``little $r$'' can usually be avoided by accepting approximate solutions that may violate optimality and the feasibility constraints by a tiny amount.
}
Effective solution bounds are necessary in order to be able to output (approximately) optimal solutions in (worst-case) polynomial time.
Without such bounds, the corresponding decision problem  is typically not known to be in \textbf{NP}, as its natural witnesses may have size superpolynomial in the input size.

Linear programs are prominent examples of optimization problems that enjoy effective solution bounds.
Their optimal solutions arise as solutions to systems of (consistent and independent) linear equations~\cite{Dantzig1963}.
By Cramer's rule, such solutions have polynomial bit length~\cite{Edmonds1967}, and thus the logarithm of their norms is polynomially bounded.
This effective solution bound plays a crucial role for the seminal result that linear programming has polynomial time algorithms \cite{Khachiyan:LP}.
For quadratic objective functions, optimal solutions continue to have a linear characterization (because the gradient of a quadratic is an affine linear map).
Correspondingly, convex quadratic programs, where one minimizes a convex polynomial of degree two over a polyhedron, have effective solution bounds and polynomial-time algorithms like linear programs~\cite{Kozlov1980}.
On the other hand, for other natural generalizations of linear programs, like semidefinite programs, the situation is drastically different and no effective solution bound is possible:
There are feasible semidefinite programs such that all feasible and even approximately feasible solutions have huge coefficients and the logarithm of their norms is exponential in the input size.\footnote{
  These examples are commonly attributed to Khachiyan;
  see~\cite{Ramana1997, ODonnell2017, Pataki2024}.
}
Therefore, the ellipsoid method does not yield a polynomial-time algorithm for semidefinite programming (even approximately), even though the corresponding separation problem has polynomial-time algorithms.
The lack of effective solution bounds is also the reason why semidefinite programming is not known to be in~\textbf{NP}~\cite{Porkolab1997, Ramana1997}.

The unconstrained minimization of convex polynomials of degree four and higher is a fundamental optimization problem whose optimal solutions have no linear characterization
(because critical points are defined by equations of degree three and higher).
At the same time, previous ``bad'' examples (like for semidefinite programs) do not apply.
In this paper, we develop new techniques to show effective solution bounds for this problem.
These techniques extend to \emph{convex polynomial programming},
where one minimizes a convex polynomial (of arbitrary degree) over a polyhedron. 
Our solution bounds, combined with the ellipsoid method,
yield the first efficient algorithm for solving convex polynomial programs at degree four and above (to an arbitrary degree of accuracy).
Previously, no efficient algorithm for computing the approximate minimum of a convex polynomial of degree four was known (even without constraints).

\subsection{Minimizers of (convex) polynomials} 

If $p$ is a polynomial of degree $2$, its critical points are solutions to the \emph{linear} system $\nabla p (x) = 0$. 
This permits us to control the norm of its minimizers, similarly to the case of linear programs. For polynomials of larger degree, this approach does not work. 
In fact, already for degree $4$, \emph{nonconvex}~polynomials may not admit an effectively bounded minimizer at all. 
Consider, e.g., 
\begin{equation} \label{EQ:exmp_doublyexp}
    p(\x) \coloneqq (x_{1}^2 - x_{2})^2 +  \ldots + (x_{n-1}^2 - x_n)^2 + (x_n - 2)^2.
\end{equation}
The minimum of $p$ over $\R^n$ is $0$, and it is attained if and only if $x_n = 2$, and $x_{i-1} = x_{i}^2$ for all $1 \leq i < n$. Thus, the unique minimizer $\x^*$ of $p$ satisfies $x^*_1 = 2^{2^{n-1}}$, meaning its norm is \emph{doubly}-exponential in~$n$.

The upshot is that, in order to establish effective solution bounds for \emph{convex} polynomials, it is imperative to exploit their convexity. 
On the other hand, it is \textbf{NP}-hard to decide whether a given polynomial is convex~\cite{Ahmadi2011}. This means we cannot rely on an algorithmic characterization of convexity for polynomials.  
Furthermore, convex polynomials do not (obviously) satisfy the type of quantitative properties typically used in the analysis of optimization algorithms. For example, even strictly convex polynomials are not \emph{strongly} convex in general (consider ${f(x) = x^4}$).

\paragraph{The Hessian determinant.}
Instead, we will work with a purely \emph{qualitative} characterization: $f$~is convex if and only if its matrix $\nabla^2 f(x)$ of second derivatives is positive semidefinite for all $x \in \R^n$. 
To translate from this qualitative description to the quantitative bounds we are after, the key is to consider the \emph{Hessian determinant} $h_f(x) \coloneqq \det 
\left(\nabla^2 f(x)\right)$ of $f$, which describes its (local) curvature.
Note that, by convexity, we have $h_f(x) \geq 0$ for all $x \in \R^n$. The Hessian determinant captures the local behavior of $f$ around a point $a \in \R^n$ via the following dichotomy:
\begin{itemize}[noitemsep]
    \item If $h_f(a) > 0$, then $f$ is \emph{locally $\mu$-strongly convex} around $a$, meaning there are $\mu, r > 0$ such that
    \[f(x) \geq f(a) + \langle \nabla f(a), x-a \rangle + \mu \cdot \|x-a\|^2 \quad \forall x \in B(a, r).
    \]
    \item If $h_f(a) = 0$, then $f$ has a \emph{local direction of linearity} at $a$, meaning there is a $v \in \R^n$ such that
    \[
        \frac{\partial^2}{\partial t^2} f(a + tv) = 0. 
    \]
\end{itemize}
\paragraph{Local to global.}
The core of our analysis is to show that both sides of the dichotomy above exhibit \emph{local-to-global} behavior.
First, we show that if the Hessian determinant is positive at even a single point $a \in \R^n$, this immediately implies a \emph{global} $\mu$-strongly convex \emph{lower bound} on~$f$. Moreover, the parameter $\mu$ can be controlled in terms of the binary encoding length of $f$ independently of $a$.
Second, we show that if~$f$ has a \emph{local} direction of linearity \emph{everywhere}, then it has a \emph{global} direction of linearity. That is, if $h_f(x) = 0$ for all $x \in \R^n$, there is $v \in \R^n$ (independent of $x$) such that
\[
    \frac{\partial^2}{\partial t^2} f(x + tv) = 0 \quad (\forall x \in \R^n),
\]
which implies that $f$ is linear along $v$.
Projecting onto the complement of this direction allows us to reduce to a problem in fewer variables. This fact has an interesting connection to work of O. Hesse himself~\cite{OHesse}, who (mistakenly) claimed it holds for general (nonconvex) polynomials.
Together, these observations will allow us to write any convex polynomial as the sum of a linear function and a polynomial (of fewer variables) that admits a strongly convex lower bound~(\Cref{THM:quantstructure}). 

\subsection{Main contributions}
As our main contribution, we show a singly-exponential upper bound on the norm of a global minimizer of a convex polynomial $f$ on a convex polyhedron $P = \{Ax \leq b\}$.
To state our precise results, we must first specify our input model. We write $\bc(q)$ for the bit length of $q \in \Q$. The bit length of a vector or matrix is the sum of the bit lengths of its entries.
We write $\enc{P} \coloneqq \bc{(A)} + \bc{(b)}$ for the (total) encoding length of $P$.
We encode a polynomial $f(x) = \sum_{\alpha} f_\alpha x^\alpha$ using a \emph{binary} representation for its (nonzero) coefficients $f_\alpha \in \Q$, and a \emph{unary} representation for the multi-indices~${\alpha \in \N^n}$. Thus, its encoding length is $\enc{f} = \Theta(n + \mathrm{deg}(f) + \sum_{\alpha: f_\alpha \neq 0} \bc(f_\alpha))$.\footnote{This encoding is more efficient than the ``naive'' encoding of $f$ as a (dense) vector of coefficients in the monomial basis (whose size scales exponentially in $\deg(f)$). Our main results \Cref{THM:mainpolyhedron} and \Cref{COR:mainpolyhedron} are novel even in this dense model; in fact already for \emph{fixed} degree $\deg(f) = 4$.}

\begin{restatable}{theorem}{THMmainpolyhedron}
    \label{THM:mainpolyhedron}
    Let $f \in \Q[x]$ be a convex polynomial.
    Let $P$ be a (nonempty) convex polyhedron.
    Then, either $f$ is unbounded from below on $P$, or it attains its minimum on $P$ at a point $x^*$ of norm 
    \[
    \log \|x^*\| \leq \poly(\enc{f}, \,\enc{P}).
    \]
\end{restatable} \noindent
As a consequence, we can minimize $f$ over $P$ in polynomial time using the ellipsoid method.
\begin{restatable}{corollary}{CORmainpolyhedron}
    \label{COR:mainpolyhedron}
    Let $f \in \Q[x]$ be a convex polynomial. Let $P$ be a (nonempty) convex polyhedron. Let~${\epsilon >0}$. We can decide in time $\poly(\enc{f},\, \enc{P},\, \log(1/\epsilon))$ whether $f$ is unbounded from below on $P$, and, if not, output a point~$\tilde x \in P$ with $f(\tilde x) \leq \min_{x \in P} f(x) + \epsilon$.
\end{restatable}
The main technical tool in our proof of~\Cref{THM:mainpolyhedron} is a structure theorem for convex polynomials, which we believe to be of independent interest. It shows that any convex polynomial can be written as the sum of a linear function and a polynomial (of fewer variables) that has a $\mu$-strongly convex \emph{lower bound}, where $\mu$ is at least inversely exponential in $\enc{f}$. 

\begin{restatable}{theorem}{quantstructurethm}
\label{THM:quantstructure}
Let $f \in \Q[x]$ be a convex polynomial. Then, there is a subspace $\mU \subseteq \R^n$, and a rational vector $w \in \mU^\perp$ (possibly $w = 0$), such that, writing $x_\mU$ for the projection of $x$ onto $\mU$,
\[
    f(x) = f(x_\mU) - \langle w, x \rangle,
\]
where $f(x_\mU)$ has a $\mu$-strongly convex lower bound on $\mU$. That is, there is a matrix $U \in \Q^{k \times n}$ whose rows are orthogonal and span $\mU$, and a $k$-variate, $\mu$-strongly convex, quadratic polynomial $q \in \Q[y_1, \ldots, y_k]$ such that $f(x_\mU) \geq q(Ux)$ for all $x \in \R^n$.
Moreover, the matrix $U$ and the vector $w$ can be computed in polynomial time in~$\enc{f}$, the lower bound~$q$ satisfies~$\enc{q} \leq \poly(\enc{f})$, and $\mu \geq 2^{-\poly(\enc{f})}$.
\end{restatable} \noindent
\Cref{THM:quantstructure} allows us to give a quantitative description of the minimizers of $f$ on $P$. For example, in the unconstrained case (where $P = \R^n)$, we have $f_{\min} = -\infty$ if and only if $w \neq 0$. If $w = 0$, the $\mu$-strongly convex lower bound on $f$ straightforwardly shows that its global minimizer is attained at a point of norm at most exponential in $\poly(\enc{f})$. For general $P$, the situation is more complicated. There, we show that $f$ is unbounded from below if and only if a certain linear system of inequalities (involving $U$ and $w$) has a solution. If not, Farkas' lemma gives us a witness for this fact, which we use to show that any minimizer of $f$ on $P$ has small norm when projected onto the subspace~${\mU \oplus \spann{w}}$.
Finally, by a lifting argument, we show that there exists a minimizer of small norm in the full space.

\subsection{Complexity landscape of polynomial programming}
A \emph{polynomial program} asks to minimize a rational polynomial $f$ over a convex polyhedron $P$ described by rational linear inequalities $Ax \leq b$; that is, to output a binary representation of
\begin{equation} \tag{POP} \label{EQ:PolynomialProgram}
    x^* = \mathrm{argmin}_{x \in \R^n} \left\{ f(x): x \in P \right\}.
\end{equation}
In this section, we discuss the complexity of this problem (in the Turing model) in the convex and general setting, and for different values of $d = \deg(f)$. Our discussion is summarized in~\Cref{TABLE:complexity}. For ease of exposition, we assume that $x^*$ is well-defined, i.e., that the minimum above is attained and unique.
We are interested in algorithms whose runtime is bounded polynomially in $n$ and the number of bits required to specify $f$ and $P$. 
Consider the associated \emph{decision problem}:
\begin{equation} \label{EQ:decision}\tag{D}
    \text{Given } (f, P), \text{ output } 
    \begin{cases}
        \textsc{yes} \quad & \text{ if } \quad \exists x \in P : f(x) \leq 0, \\
        \textsc{no} & \text{ if } \quad \forall x \in P : f(x) > 0.
    \end{cases}
\end{equation}
\paragraph{Compact witnesses.} A \emph{witness} for~\eqref{EQ:decision} is a point $w \in P$ satisfying $f(w) \leq 0$. Clearly, the answer to~\eqref{EQ:decision} is \textsc{yes} if and only if there exists a witness. Furthermore, if there exists a witness $w \in \Q^n$ of polynomial bit length (in the size of the input), then this \textsc{yes}-answer can be verified efficiently by evaluating $f(w)$ and the inequalities $Aw \leq b$. We call such a witness \emph{compact}.
Thus, the existence of a compact witness for every \textsc{yes}-instance of~\eqref{EQ:decision} implies that the problem lies in \textbf{NP}. We have already seen that this is the case for linear programs ($d=1$) and for quadratic programs~($d=2$)~\cite{Vavasis1990}.
In both cases, the ``canonical'' witness $w = x^*$ happens to be compact.\footnote{For linear programs, $x^*$ is a vertex of $P$, and thus the solution to a subsystem of $Ax = b$ \cite{Dantzig1963, Edmonds1967}. 
For unconstrained quadratic programs ($P = \R^n$), $x^*$ is a critical point of $f$, and thus a solution to the linear system $\nabla f(x) = 0$.
In the general setting, it turns out $x^*$ is still the solution to a linear system, involving KKT(-like) conditions~\cite{Kozlov1980, Vavasis1990}.} 
As a direct result, linear programming and \emph{convex} quadratic programming lie in \textbf{P}, as $x^*$ can be identified efficiently by the ellipsoid method~\cite{Khachiyan:LP, Kozlov1980}. On the other hand, (nonconvex) quadratic programming is~\textbf{NP}-complete (while $x^*$ is compact, it might be hard to find~\cite{Sahni1974}).

For $d \geq 4$, the minimizer $x^*$ is typically not compact. In fact, there might not exist any compact witness for~\eqref{EQ:decision}. For example, the global minimum of $f(x) = (x^2-2)^2$ over $P = \R$ is $0$, but the only witnesses of this fact are $\pm \sqrt{2} \not \in \Q$. 
Even if we assume $f$ is convex, there are examples of degree~$4$ where $x^* \not \in \Q^n$, and examples of degree $6$ where no rational witnesses exist  at all (see~\Cref{APP:Complexity}). 

\paragraph{Effective solution bounds.}
If the canonical witness $x^*$ is not rational (compact), this means that the polynomial program~\eqref{EQ:PolynomialProgram} \emph{cannot} be solved (efficiently), as it is not possibly to output a binary representation of~$x^*$ (in polynomial time). 
This appears a somewhat artificial limitation. 
Instead, one could ask to output an \emph{approximate} minimizer $\tilde x \in P$, satisfying $f(\tilde x) \leq f(x^*) + \epsilon$, where $\epsilon > 0 $ is part of the input. This corresponds to the following \emph{promise problem}:
\begin{equation} \label{EQ:promise}\tag{P}
    \text{Given } (f, P, \epsilon), \text{ output } 
    \begin{cases}
        \textsc{yes} \quad & \text{ if } \quad \exists x \in P : f(x) \leq 0, \\
        \textsc{no} & \text{ if } \quad \forall x \in P : f(x) > \epsilon.
    \end{cases}
\end{equation}
A witness for~\eqref{EQ:promise} is a point $\tilde w \in P$ with $f(\tilde w) \leq \epsilon$. Existence of such a $\tilde w$ is equivalent to \textsc{yes} being a valid answer to~\eqref{EQ:promise}. As before, the existence of a compact witness for each \textsc{yes}-instance implies that~\eqref{EQ:promise} is in \textbf{NP}.
The upshot is that any witness $w$ for~\eqref{EQ:decision} whose \emph{norm} is exponentially bounded in (a polynomial of) the size of the input can be rounded to a \emph{compact} witness $\tilde w$ for~\eqref{EQ:promise}. 
Thus, to show that \eqref{EQ:promise} is in \textbf{NP}, it suffices to show that $x^*$ has exponentially bounded norm. 

Notably, polynomial programs generally do not admit an exponentially bounded solution, already when $d=4$. Indeed, the quartic polynomial defined in~\eqref{EQ:exmp_doublyexp} has a unique global minimizer of norm $2^{2^{n-1}}$ on $P = \R^n$. Our main contribution (\Cref{THM:mainpolyhedron}) is to show that \emph{convex} polynomials (of any degree) \emph{do} admit an exponentially bounded minimizer. Contrary to the degree-$2$ case, the existence of compact (approximate) witnesses for $d\geq 4$ thus depends on convexity. Using the ellipsoid method, this shows that (approximate) convex polynomial programming is in \textbf{P}~(\Cref{COR:mainpolyhedron}). 

\begin{table}[ht]
\centering
\begin{tabular}{lll|ll}
\toprule
\makebox[3cm][l]{\textbf{objective}} & \makebox[3cm][l]{\textbf{compact wit.}~\eqref{EQ:decision}} & \makebox[3cm][l]{\textbf{compact wit.}~\eqref{EQ:promise}} & \makebox[2.5cm][l]{\textbf{complexity}~\eqref{EQ:decision}} & \textbf{complexity}~\eqref{EQ:promise}\\
\midrule
linear& yes$^{a}$~{\scriptsize\cite{Dantzig1963,Edmonds1967}}& --- & \textbf{P} {\scriptsize\cite{Khachiyan:LP}} & --- \\
convex quadratic & yes$^{a}$ {\scriptsize\cite{Kozlov1980}}& --- & \textbf{P}~{\scriptsize\cite{Kozlov1980}} & --- \\
quadratic & yes$^{a}$ {\scriptsize\cite{Vavasis1990}}& --- & --- & \textbf{NP}-hard {\scriptsize\cite{Sahni1974}} \\
convex quartic & unknown$^{b}$ & yes~ {\scriptsize (Thm. \ref{THM:mainpolyhedron})} & unknown & \textbf{P}~{\scriptsize (Cor. \ref{COR:mainpolyhedron})}  \\
 quartic & --- & \makebox[0.65cm][l]{no}~{\scriptsize (Eq. \eqref{EQ:exmp_doublyexp})} & --- & \textbf{NP}-hard$^{c}$ \\
convex $d \geq 6$ & no~{\scriptsize (App. \ref{APP:Complexity})} & yes~ {\scriptsize (Thm. \ref{THM:mainpolyhedron})} & unknown & \textbf{P}~{\scriptsize (Cor. \ref{COR:mainpolyhedron})} \\
\bottomrule
\end{tabular}
\caption{Complexity and availability of compact witnesses for problems~\eqref{EQ:decision} and~\eqref{EQ:promise}, corresponding to exact and approximate polynomial programming, respectively. Omitted entries (``---'') follow directly from adjacent ones (e.g., a compact witness for~\eqref{EQ:decision} is also a compact witness for~\eqref{EQ:promise}).
{ \footnotesize
\\ $^a$the canonical witness $w = x^*$ is compact. 
\\ $^b$the canonical witness $w = x^*$ is not compact in general, but a compact witness might exist. See~\Cref{APP:Complexity}.
\\ $^c$the problem is \textbf{NP}-hard even in the unconstrained setting ($P = \R^n$)~\cite{Laurent2008}.
}}
\label{TABLE:complexity}
\end{table}

\paragraph{Local polynomial optimization.}
Even \emph{local} optimization tasks involving polynomials are typically hard:
Determining whether a polynomial of degree $4$ has a local minimum or a critical point is hard; it is even hard to decide if a given (critical) point \emph{is} a local minimum~\cite{Ahmadi2022}. Determining an (approximate) local minimizer of a quadratic on a polytope is also hard~\cite{Ahmadi2022a}.

\paragraph{Real models of computation.}
In real models of computation (e.g., in the \textit{Blum-Schub-Smale} model~\cite{BSS1989} or in the \emph{real RAM} model), the decision problem~\eqref{EQ:decision} is always in \textbf{NP} (as evaluating $f(x)$ and $Ax \leq b$ at any point $x \in \R^n$ is easy there). However, in these models, it is not even known whether \emph{linear} programs can be solved in polynomial time~\cite{Smale91}
(as the runtime of an algorithm is no longer allowed to depend on the size of the entries of $A$ and $b$). 

\subsection{Related work}
\label{SEC:relatedwork}
\paragraph{Higher-order Newton methods.}
The global minimization of convex polynomials comes up naturally in Newton's method for optimization. This method attempts to find the minimizer of a smooth function $g : \R^n \to \R$ by successively approximating it by its second-order Taylor expansion around an iterate $x^k \in \R^n$, setting $x^{k+1}$ to be a global minimizer of this approximation. 
If the Hessian matrix $\nabla^2 g (x^k) \succeq 0$ is positive semidefinite at $x^k$, then the Taylor expansion around $x^k$ is a convex quadratic, and its critical points are global minimizers.
Newton's method converges quadratically within a small basin around local minimizers where the Hessian of $g$ is positive definite and locally Lipschitz.
It has been extensively studied whether faster convergence can be achieved by considering Taylor expansions of~$g$ of higher order, see e.g.,~\cite{Nesterov2019,Ahmadi2024}. A difficulty is that the degree-$d$ expansion of $g$ around~$x^k$ is typically nonconvex when $d > 2$ (even if $\nabla^2 g(x^k) \succeq 0$), and it is not clear how to compute its global minimizer (if it even exists). 
To avoid this problem, the authors of~\cite{Nesterov2019} and \cite{Ahmadi2024} add a regularizing term that ensures convexity. 
Still, efficient optimization of the resulting convex polynomial then relies on additional properties inherited from the specifics of the regularization scheme: in~\cite{Nesterov2019}, it satisfies a \emph{relative smoothness condition}; in~\cite{Ahmadi2024} it is \emph{sum-of-squares-convex} (see below). 
This raises the question whether such additional properties are necessary.\footnote{In fact, this question is raised already by the author in~\cite{Nesterov2019}, who states that they were unable to locate any method in the literature aimed at solving the general problem of minimizing a convex, multivariate polynomial.} 
Our work shows that convexity alone is enough.

\paragraph{Complexity of SDP.}
Semidefinite programming (SDP) is a natural extension of linear programming with important applications in (combinatorial) optimization. 
This includes the Goemans-Williamson approximation for max-cut~\cite{GW1995}, and an efficient algorithm to compute the independence number of a perfect graph~\cite{Schrijver:Ellipsoid}.
Under additional assumptions, SDPs can be solved in polynomial time using the ellipsoid method~\cite{Schrijver:Ellipsoid} or interior-point methods~\cite{Klerk2016}.
However, the computational complexity of general SDP is not known. The reason is that SDPs do not admit effective solution bounds~\cite{Porkolab1997, Pataki2024}; consider the  following example due to Khachiyan:
\begin{equation} \label{EQ:expSDP}
   \exists x \in \R^n :
\begin{bmatrix}
x_{1} & x_{2} \\
x_{2} & 1
\end{bmatrix} 
\succeq 0, \, \ldots \, , 
\begin{bmatrix}
x_{n-1} & x_{n} \\
x_{n} & 1
\end{bmatrix} \succeq 0,
\,
x_n \geq 2?
\end{equation}   
Any solution to this problem must satisfy~$x_1 \geq 2^{2^{n-1}}$, meaning it has doubly-exponential norm. 
The best available results are that, in the Turing model, SDP is either in $\mathbf{NP} \cap \mathbf{coNP}$ or not in $\mathbf{NP} \cup \mathbf{coNP}$~\cite{Ramana1997};  in the real (Blum-Shub-Smale) model, it is in $\mathbf{\textbf{NP}} \cap \mathbf{coNP}$~\cite{Ramana1997}.

\paragraph{Effective solution bounds for sum-of-squares relaxations.} 
There is a well-developed literature on convex relaxations that produce tractable bounds on the (global) minimum $p_{\min}$ of a polynomial~$p$. 
These are based on the observation that 
$
    p_{\min} = \max_{\lambda \in \R} \{ \lambda : p(x) - \lambda \geq 0 \, \forall x \in \R^n\},
$
which transforms a (nonconvex) polynomial optimization problem into a (convex) problem over the cone $\mathcal{P}$ of nonnegative polynomials. Testing membership in $\mathcal{P}$ is \textbf{NP}-hard, but one may obtain a tractable \emph{lower bound} on $p_{\min}$ by optimizing over a subcone of $\mathcal{P}$ that admits an efficient separation oracle. 
The most famous example is the cone $\Sigma \subseteq \mathcal{P}$ of \emph{sums of squares (sos)}, i.e., polynomials of the form $\sigma(x) =  s_1(x)^2 + \ldots + s_\ell(x)^2$, $s_i \in \R[x]$. It is possible to optimize over $\Sigma$ via a semidefinite program. 
This approach can be generalized to constrained problems, leading to a sequence of lower bounds called the \emph{sum-of-squares hierachy}~\cite{Parrilo2000, Lasserre2001}. Algorithms based on the sos hierarchy have had a significant impact in theoretical computer science~\cite{BarakSteurer2014} . 
However, its precise computational complexity is not known: Pathological examples similar to~\eqref{EQ:expSDP} can be realized as the sos relaxation of a (constrained) polynomial optimization problem~\cite{ODonnell2017}. 
Recent work shows that polynomial-time computability of the sos hierarchy can be guaranteed under additional assumptions~\cite{Raghavendra2017,Gribling2023, Bortolotti2025, Bortolotti2025a}. Essentially, these works show that certain special classes of semidefinite programs admit effectively bounded solutions.

\paragraph{Sum-of-squares-convexity.}
The interplay between convex polynomials and sums of squares has also been studied.
A polynomial $p \in \R[x]$ is called \emph{sum-of-squares-convex} if the $2n$-variate polynomial $\langle y, \nabla^2 p(x) \cdot y \rangle$
is a sum of squares of polynomials in $\R[x ,y]$. This implies that~${\nabla^2 p(x) \succeq 0}$ for all $x \in \R^n$. 
Thus, sos-convex polynomials are convex in the regular sense, and their convexity has an algebraic proof. (But, not all convex polynomials are sos-convex~\cite{Ahmadi2013}.) 
Contrary to regular convexity (which is \textbf{NP}-hard to detect~\cite{Ahmadi2011}), sos-convexity of a polynomial can be determined by testing feasibility of a semidefinite program~\cite{Ahmadi2024}. 
Moreover, the (global) minimum of an sos-convex polynomial can be determined by solving an SDP. Indeed, the first level of the sum-of-squares hierarchy is exact for sos-convex polynomials (even under sos-convex constraints)~\cite{LasserreConvex}. In light of the discussion on sum-of-squares relaxations above, we remark that (to the best of the authors' knowledge) it has never been examined whether these semidefinite programs can be solved in polynomial time.

\section{Technical overview}
\label{SEC:TO}
At a high level, we are concerned with understanding the quantitative behavior of the minimizers of a convex polynomial $f \in \Q[x]$ on a convex polyhedron $P = \{Ax \leq b\} \subseteq \R^n$.
In particular, we will show that $f$ is either unbounded from below on~$P$, in which case this can be detected efficiently, or~$f$ attains its minimum on $P$ at a point $x^* \in P$ whose norm is at most exponential in the encoding lengths of $f$ and $P$. 
In the latter case, the ellipsoid method can (approximately) identify this point in polynomial time. 
Together, this yields our main results~\Cref{THM:mainpolyhedron} and~\Cref{COR:mainpolyhedron}.

The main technical tool for establishing this dichotomy is a structure theorem  for convex polynomials (\Cref{THM:quantstructure}), which allows us to write $f$ as the sum of a linear function and a polynomial (of fewer variables) which is bounded from below by a $\mu$-strongly convex quadratic. 
Namely, we show that any convex polynomial $f$ decomposes as
\begin{equation} \label{EQ:TO:structurethm}
    f(x) = f(x_\mathcal{U}) - \langle w, x \rangle \geq q(Ux) - \langle w, x \rangle \quad (\forall x \in \R^n),
\end{equation}
where $U \in \Q^{k \times n}$ is a matrix with orthogonal rows spanning a subspace $\mathcal{U} \subseteq \R^n$, $w \in \mU^\perp$ is a rational vector (possibly $w = 0$), and $q \in \Q[y]$ is a $k$-variate, $\mu$-strongly convex quadratic.
Moreover, we show that $U, w$ can be computed in polynomial time in $\enc{f}$, and that $\mu \geq 2^{-\poly(\enc{f})}$.

In the unconstrained case (when $P = \R^n$), this immediately gives us the desired norm bound on the minimizer of $f$. Indeed, we then have $f_{\min} = - \infty$ if and only if $w \neq 0$, in which case $\lim_{\lambda \to \infty} f(\lambda w) = - \infty$. On the other hand, if $w = 0$, the strongly convex lower bound~\eqref{EQ:TO:structurethm} on~$f$ implies that it has a global minimizer of small norm (in terms of $\mu$). 
For general $P$, rather than checking whether $w = 0$, we will show that $f$ is unbounded from below on $P$ if and only if a certain linear system of inequalities (involving $U$ and $w$) is feasible. If not, Farkas' lemma provides an explicit witness of this fact, which we use to show that $f$ attains its minimum on $P$ at a point whose norm is at most singly-exponential in $\enc{f}$ and $\enc{P}$.
We explain this in more detail in~\Cref{SEC:applystruct} below. First, we outline our proof of~\Cref{THM:quantstructure}, i.e., of the decomposition~\eqref{EQ:TO:structurethm}.

\subsection{Structure theorem for convex polynomials}
The key to our proof of~\Cref{THM:quantstructure} is to consider the \emph{Hessian determinant} (or just \emph{hessian}) of $f$, which is given by $h_f(x) \coloneqq \det \left( \nabla^2 f(x) \right)$. By convexity, $h_f(x) \geq 0$ for all $x \in \R^n$. Intuitively, the hessian captures whether the graph of $f$ is locally `curved' around a point $x \in \R^n$ ($h_f(x) > 0$), or `straight' in at least one direction~(${h_f(x) = 0}$), corresponding to a zero eigenvector of $\nabla^2 f(x)$. For our proof, we exhibit two complimentary local-to-global phenomena involving the hessian of $f$:
\begin{enumerate}
    \item If $h_f \not \equiv 0$, i.e., if $\nabla^2 f(a) \succ 0$ at even a single point $a \in \R^n$, then $f$ admits a \emph{global}, \mbox{$\mu$-strongly} convex lower bound of degree two. Moreover, we have $\mu \geq 2^{-\poly(\enc{f})}$ independently of~$a$.
    \item If $h_f \equiv 0$, i.e., if $\nabla^2 f(x)$ is singular for all $x \in \R^n$, then $f$ admits a \emph{global} direction of linearity. That is, there is a nonzero $v \in \R^n$ such that the directional derivative $\mynabla{v} f$ is constant on~$\R^n$.
    In other words, we have $f(x) = f(x_{\mV^\perp}) + c \langle x, v \rangle$ for all $x \in \R^n$, where $c \in \R$ and $\mV = \spann{v}$.
\end{enumerate}
At a high level, these observations allow us to prove~\Cref{THM:quantstructure} as follows: if $h_f \not \equiv 0$ is not identically zero, we are done immediately (setting $\mU = \R^n$ and $w = 0$). If $h_f \equiv 0$, we may conclude that $f$ has a direction of linearity. After projecting onto the complement of that direction, we are left with a polynomial in one fewer variable which differs from $f$ only by a linear term.
Iterative application of this idea allows us to ``eliminate'' variables one at a time until the ``remaining'' polynomial has a nonzero hessian (and thus admits a strongly convex lower bound). For technical reasons, it turns out that it is preferable to eliminate all variables ``in one shot'', and we will show how to do so.
In what follows, we explain in more detail how to prove the two observations above separately, and how to combine them to obtain the decomposition~\eqref{EQ:TO:structurethm}.

\paragraph{Nonzero hessian implies a strongly convex lower bound.} We show that any convex polynomial $f \in \Q[x]$ with $h_f \not \equiv 0$ admits a $\mu$-strongly convex, quadratic lower bound, where $\mu$ is at least inverse exponential in the encoding size of $f$. Our starting point is the following lemma.
\begin{restatable}[{\cite[Lemma 5]{Ahmadi2024}}]{lemma}{AhmadiLB}
\label{LEM:stronglyconvex}
    Let $p \in \R[x]$ be a convex polynomial. For any $a \in \R^n$, we have
    \begin{equation} 
    \label{EQ:stronglyconvexAhmadi}
        p(\x) \geq p(a) + \langle \nabla p (a), 
        \, \x-a \rangle + \frac{\langle x-a, \, \nabla^2 p(a) \cdot (x-a) \rangle}{4 \deg(p)^2} \quad (\forall x \in \R^n).
    \end{equation}    
\end{restatable}
This lemma can be seen as an extension of the fact that a convex function is globally lower bounded by its tangents: The RHS of~\eqref{EQ:stronglyconvexAhmadi} is equal to the second-order Taylor expansion of $p$ around~$a$, up to a rescaling of the quadratic term by a factor $\approx 1/\deg(p)^2$. Its proof relies on a clever combination of basic facts on integration of low-degree (matrix) polynomials: Write the first-order error term of the Taylor series as an integral involving $\nabla^2 p$; apply a classical quadrature rule~\cite{ClenshawCurtis1960} to reduce the integral to a finite, weighted sum $\sum_{i}w_i \cdot \nabla^2 p(y^{(i)})$ with weights $\approx 1/\deg(p)^2$~\cite{Imhof63}; finally, use the fact that $\nabla p(x) \succeq 0$ for all $x$ to lower bound the sum by a single term $w_0 \cdot \nabla^2 p(a)$.

From~\eqref{EQ:stronglyconvexAhmadi}, we get a strongly convex lower bound on $f$ if $h_f \not \equiv 0$. Namely, if $\nabla^2 f(a) \succ 0$, then
    \begin{equation} \label{EQ:stronglyconvex}
        f(\x) \geq q(x) \coloneqq f(a) + \langle \nabla f (a), \x-a \rangle + \frac{\lambda_{\min}\left(\nabla^2 f (a)\right)}{4 \deg(f)^2} \cdot \|x-a\|^2  \quad (\forall x \in \R^n).
    \end{equation}
This lower bound was used in~\cite{Ahmadi2024} to establish that convex polynomials with nonzero hessian are \emph{coercive}, and therefore have a unique global minimizer. For us, this is not enough; we need an explicit bound on the norm of that minimizer. 
Note that~\eqref{EQ:stronglyconvex} actually gives us a $\mu$-strongly convex lower bound on $f$, where
\begin{equation}
    \label{EQ:mu}
    \mu = \mu(a) = \frac{\lambda_{\min}(\nabla^2 f (a))}{2 \deg(f)^2} >0.
\end{equation}
Unfortunately, this $\mu$ could be arbitrarily close to $0$, as there is no way to control $\lambda_{\min}(\nabla^2 f (a))$. 
We would need to find an $a \in \R^n$ such that $\mu(a)$ in~\eqref{EQ:mu} is at least inversely exponential in $\enc{f}$. Importantly, we need to do so without making any further assumptions on $f$. We show that this is possible by establishing the following two facts:
\begin{itemize}
    \item If $h_f \not \equiv 0$, then there exists an $a \in \Q^n$ such that $\nabla^2 f(a) \succ 0$ with $
    \bc(a) \leq \poly(n,\, \deg(f))$. That is, if $h_f(a) \neq 0$ for \emph{any} $a \in \R^n$, then in fact $h_f(a) \neq 0$ for an $a \in \Q^n$ of small bit length. 
    \item For any $a \in \Q^n$ with $\nabla^2 f(a) \succ 0$, we have 
    $
    \lambda_{\min}\big( \nabla^2 f(a) \big) \geq 2^{-\poly(\enc{f}, \, \bc(a))}.$ That is, for any $a$ of small bit length, the lower bound~\eqref{EQ:stronglyconvex} is $\mu$-strongly convex with $\mu > 0$ not too small.
\end{itemize}
Combined, these facts show that
\begin{equation} \label{EQ:TO:stronglyconvexlb}
    h_f \not \equiv 0 \implies f \text{ admits a $\mu$-strongly convex, quadratic lower bound, where } \mu \geq 2^{-\poly(\enc{f})}.
\end{equation}

To prove the first fact, note that the set of (real) zeroes of a nonzero polynomial has measure zero in~$\R^n$. Thus, if $h_f \not\equiv 0$, we intuitively expect virtually any point $a \in \Q^n$ to satisfy $\nabla^2 f(a) \succ 0$. This intuition can be made precise using the Schwartz-Zippel lemma (\Cref{LEM:ZS}), which gives an explicit set of points $B \subseteq \Z^n$, each of small bit length, such that $h_f(a) \neq 0$ for at least one $a \in B$.

For the the second fact, note that $\nabla^2 f(a)$ is a positive definite, rational matrix whose bit length is bounded by $\poly(\enc{f},\, \bc(a))$. The eigenvalues of such matrices can be bounded away from $0$ using Cramer's rule (see~\Cref{LEM:integermatrix} and \Cref{COR:rationalmatrix}).

\paragraph{Identically zero hessian implies direction of linearity.}
In light of the above, we investigate when the hessian~$h_f$ of a convex polynomial $f$ is identically zero, i.e., when $\nabla^2 f$ is \emph{nowhere} definite. We show that this happens if, \emph{and only if} $f$ has a (global) direction of linearity. That is to say,
\begin{equation}
    \label{EQ:Hessianvanishes}
    h_f \equiv 0 \iff \exists v \in \R^n \setminus \{ 0\}, \, c \in \R : \mynabla{v} f \equiv c.
\end{equation}
Note that the backward implication of~\eqref{EQ:Hessianvanishes} is immediate: if there exist (fixed) $v \in \R^n$, $c \in \R$ so that $\mynabla{v} f(x) = c$ for all $x \in \R^n$, then~$v$ is a kernel vector of $\nabla^2 f(x)$ for all $x \in \R^n$. The forward implication tells us that we may reverse the quantifiers: If $f$ has a \emph{local} direction of linearity for every~${x \in \R^n}$ (i.e, $\nabla^2 f(x)$ has a kernel vector $v = v(x)$ for all $x \in \R^n$), then in fact it has a \emph{global} direction of linearity.
Its validity is more subtle; in particular it \emph{does not} hold for nonconvex polynomials. Interestingly, O. Hesse himself worked on this question, and (mistakenly) claimed the equivalence holds in the general case. We discuss this further below.

To prove the forward implication of~\eqref{EQ:Hessianvanishes}, assume that $h_f \equiv 0$. We view the gradient~$\nabla f$ of $f$ as the smooth function~${\R^n \to \R^n}$ given by $x \mapsto \nabla f(x)$, and consider its range 
\[
\nabla f(\R^n) \coloneqq \{ \nabla f(x) : x \in \R^n \} \subseteq \R^n.
\]
We can deduce the following two properties of this set:
\begin{itemize}
    \item The Hessian $\nabla^2 f$ may be seen as the Jacobian of $\nabla f$. The fact that~${\det \big(\nabla^2 f\big) \equiv 0}$ means that every $x \in \R^n$ is a critical point of $\nabla f$. Thus,
    Sard's theorem (\Cref{THM:Sard}) tells us that~$\nabla f(\R^n)$ has Lebesgue measure zero in~$\R^n$.
\item The range of the gradient of a (sufficiently smooth) convex function is known to be an \emph{almost convex} set, meaning $\nabla f(\R^n)$ contains the interior of its convex hull.\footnote{This fact appears to be rather classical, going back at least to~\cite[Corollary 2]{Minty1964} (which proves a more general statement). For completeness, we include a short proof using convex conjugates (\Cref{LEM:convexgradient}). 
} 
\end{itemize}
It follows that $\nabla f(\R^n)$ is contained in an affine hyperplane: If not, its convex hull would contain a full-dimensional simplex, and its interior would have positive measure. But that means that there is a $v \in \R^n$ and a $c \in \R$ such that~${\mynabla{v} f(x) \coloneqq \langle v, \nabla f(x) \rangle = c}$ for all $x \in \R^n$, meaning $f$ has a (global) direction of linearity.

\paragraph{One-shot variable elimination: proof of the structure theorem.} \label{SEC:proofdecomp}
Next, we explain how to combine the above
to prove our structure theorem. Recall that we wish to write a convex polynomial $f \in \Q[x]$ as the sum of a linear function and a polynomial (in fewer variables) that is lower bounded by a strongly convex quadratic. If the hessian $h_f \not \equiv 0$ is not identically zero, we are done immediately by~\eqref{EQ:TO:stronglyconvexlb}. The idea is that this assumption is essentially without loss of generality in light of~\eqref{EQ:Hessianvanishes}, which tells us that if $h_f \equiv 0$, we can identify a direction in which $f$ is linear. Projecting onto the complement of this direction allows us to reduce to a problem in one fewer variable. 
Successive application of this procedure yields a decomposition of $f$ into a linear function and a polynomial~$\hat f$ whose hessian is not identically zero (to which we may then apply~\eqref{EQ:TO:stronglyconvexlb}). 
In what follows, we formalize this idea. In fact, we will show how to eliminate all variables ``in one shot''.

First, we identify all directions of linearity of $f$, which make up a subspace $\mL \subseteq \R^n$. For this, note that the partial derivatives $\partial f / \partial x_i$ of $f$ are polynomials (in $x$). Thus, we can adopt a new way of viewing the gradient $\nabla f$, namely as the linear operator~${\R^n \to \R[x]}$ given by
\[
    \nabla f : v \mapsto \mynabla{v} f = \sum_{i=1}^n v_i \cdot \frac{\partial f}{\partial x_i} \in \R[x] \quad (v \in \R^n).
\]
Then, $\mL \subseteq \R^n$ is just the subspace of vectors mapped to a constant polynomial by $\nabla f$, i.e.,
\[
     \mL \coloneqq \{ v \in \R^n : \mynabla{v} f \equiv c \text{ for some } c \in \R \}.
\]
Clearly, $\mL \supseteq \ker( \nabla f)$.  
In fact, since $\mL$ is the inverse image of a subspace of dimension $1$, we find that $\mL = \ker(\nabla f) \oplus \spann{w}$, where $w \in \ker(\nabla f)^\perp$ is either zero or satisfies $\mynabla{w} f \equiv -1$.
Now, set $\mU = \mL^\perp$. Construct a matrix $U \in \R^{k \times n}$ whose rows $U_1, \ldots, U_k$ form an orthonormal basis of $\mU$. Let $\hat f \in \R[y_1, \ldots, y_k]$ be the polynomial defined by
$\hat f(y) \coloneqq f(\sum_{i=1}^k y_iU_i)$, so that~${\hat f (Ux) = f(x_\mU)}$.
By construction, 
\begin{equation}
\label{EQ:Hessiandecomp}
  f(x) = f(x_\mU + x_{\mL}) = f(x_\mU) - \langle w, x \rangle = \hat f (Ux) - \langle w, x \rangle \quad (\forall x \in \R^n).  
\end{equation}
We claim that the hessian of $
\hat f$ is not identically zero. By~\eqref{EQ:Hessianvanishes}, it suffices to show that $\hat f$ has no directions of linearity.
For this, note that the partial derivatives of $\hat f$ are linear combinations of the partial derivatives of $f$ along directions in $\mU$. Thus, if $\hat f$ had a direction of linearity, this would imply that~$f$ has a direction of linearity in $\mU$, which contradicts the fact that $\mU = \mL^\perp$.

\paragraph{Bit complexity.}
We have shown how to construct a matrix $U \in \Q^{k \times n}$, a polynomial~$\hat f$ with nonzero hessian, and a vector $w$ with $Uw = 0$ so that $f(x) = \hat f(Ux) - \langle w, x \rangle$. As we noted, we may then apply~\eqref{EQ:TO:stronglyconvexlb} to~$\hat f$ to obtain the desired decomposition of $f$ into a linear function and a polynomial with a $\mu$-strongly convex, quadratic lower bound.
A technical but important detail that we have ignored here is that in order to fully prove~\Cref{THM:quantstructure}, we need to make sure that this $\mu$ is sufficiently large, namely at least inverse exponential in $\enc{f}$. For that, it would suffice to show that $\enc{\hat{f}}$ is at most polynomial in $\enc{f}$.
Moreover, we need to compute the matrix~$U \in \Q^{k \times n}$ and the vector~$w\in\Q^n$ in polynomial time in $\enc{f}$. In short, we need to show that the linear algebra used to construct $U, w$ and $\hat f$ above can be performed in polynomial time (in the bit model).

The important observation is that, for $f \in \Q[x]$, the linear operator $\nabla f : v \mapsto \mynabla{v} f$ can be expressed as a rational matrix $\mat{\nabla f}$, whose columns represents the (monomial) expansion of the partial derivatives $\partial f / \partial x_i \in \Q[x]$. 
This matrix has bit length at most $\poly(\enc{f})$. 
Let $\mathds{1}$ denote the expansion of the constant polynomial $x \mapsto 1$. 
It suffices to perform the following operations:
\begin{enumerate}[noitemsep]
    \item Find a $w \in \ker\left(\mat{\nabla f}\right)^\perp$ with $\mat{\nabla f} \cdot w = -\mathds{1}$, or set $w=0$ if no solution exists;
    \item Compute an orthogonal basis for $\mU = \mL^{\perp}$, where $\mL \coloneqq \ker \left( \mat{\nabla f} \right) \oplus \spann{w}$.
\end{enumerate}
Both can be achieved in polynomial time via standard results on linear system solving and Gram-Schmidt orthogonalization in the bit model, see~\Cref{APP:LA}.
A detail here is that ortho\emph{normalization} of a set of vectors is typically not possible in polynomial time (the resulting vectors might not be rational). For this reason, the matrix $U$ appearing in~\Cref{THM:quantstructure} has merely orthogonal rows.

\paragraph{Some remarks on scope: Hesse's redemption.} \label{SEC:OHesse}
One might wonder if the decomposition~\eqref{EQ:Hessiandecomp} of a polynomial into a linear part and a part with nonzero hessian holds for \emph{nonconvex} polynomials as well. We only use convexity at one point in its proof, namely to show that $h_f \equiv 0$ if and only if $f$ has a direction of linearity~\eqref{EQ:Hessianvanishes}. Intuitively, it seems plausible that this equivalence holds in general.
It turns out that it does not, and the question has an interesting history. In particular, O. Hesse himself (mistakenly) made this claim~\cite{OHesse}, but was later corrected~\cite{Gordan1876}; 
see also~\cite{Garbagnati2009}. Without convexity, we are only able to conclude from $h_f \equiv 0$ that the partial derivatives of $f$ are \emph{algebraically dependent}, i.e., there exists a nonzero polynomial $\pi$ such that $\pi(\partial f / \partial x_1, \ldots, \partial f / \partial x_n) \equiv 0$. There are examples where this $\pi$ has to be of degree at least $2$~\cite{Gordan1876}. For completeness, we include such an example in~\Cref{APP:HessianCounterexample}. On the other hand, Hesse's claim finds some redemption in this work: it holds for convex polynomials, and this is a crucial ingredient of our proof.

A further subtlety here is that our proof of~\eqref{EQ:Hessianvanishes} does not actually make use of the fact that $f$ is a polynomial (although, we do use this fact to \emph{compute} the directions of linearity). Indeed, any (sufficiently smooth) convex function can be written as the sum of a linear function and a function whose Hessian matrix is not everywhere singular. 
On the other hand, our primary motivation for deriving this decomposition was to then apply~\eqref{EQ:TO:stronglyconvexlb} to obtain a strongly convex, quadratic lower bound on the nonlinear part. 
Such a bound cannot be obtained in general.
For example, $t \mapsto \exp(t)$ is convex, with \emph{everywhere} definite Hessian, but it is not even coercive.

\subsection{Applications in convex polynomial programming} \label{SEC:applystruct}
In this section, we show how to apply our structure theorem (\Cref{THM:quantstructure}) to convex polynomial programming.
Let $f \in \Q[x]$ be a convex polynomial, and let $P = \{Ax \leq b\}$ be a (nonempty) convex polyhedron. 
Our goal is to show that $f$ is either unbounded from below on~$P$ (in which case we show that this can be determined efficiently) or it has a minimizer $x^*$ on $P$ of norm at most $2^{\poly(\enc{f}, \, \enc{P})}$ (which can then be identified efficiently by the ellipsoid method). Together, this yields~\Cref{THM:mainpolyhedron} and~\Cref{COR:mainpolyhedron}.
Recall that the structure theorem gives us a decomposition
\begin{equation}\label{EQ:structuretheoremforapplication}
    f(x) = f(x_\mU) - \langle w,x \rangle \quad (\forall x \in \R^n),
\end{equation}
where $\mU \subseteq \R^n$ is a subspace spanned by the (orthogonal) rows of a matrix $U \in \Q^{k \times n}$, ${w \in \mU^\perp}$, and  $f(x_\mU) \geq q(Ux)$ for some $k$-variate, $\mu$-strongly convex quadratic, polynomial $q \in \Q[y]$.
Furthermore, the bit lengths of $U$, $w$ and $q$ are all polynomial in~$\enc{f}$.

\paragraph{Detecting unboundedness.}
We first prove that $f$ is unbounded from below if and only if there exists a solution to a certain linear system of inequalities, which will allow us to detect unboundedness efficiently.
Recall that in the unconstrained case (i.e., when $P = \R^n$), $f$ is unbounded from below if and only if $w \neq 0$.
In the constrained case, we prove that unboundedness is equivalent to the feasibility of the following linear system of inequalities:
\begin{equation}\label{EQ:conditionforunboundedness}
    A x^0 \leq 0, \: Ux^0 = 0, \: \langle w, x^0 \rangle = 1.
\end{equation}
If $P = \R^n$, \eqref{EQ:conditionforunboundedness} reduces to $\exists x^0 : \: Ux^0 = 0, \: \langle w, x^0 \rangle = 1$ or in other words $w \neq 0$.
Recall that, if $P = \R^n$ and $w \neq 0$, $f$ is unbounded from below since ${\lim_{\lambda \to \infty} f(\lambda w) = - \infty}$.
For general $P$, the points $\lambda w$ need not be in $P$. However, for any solution $x^0$ to \eqref{EQ:conditionforunboundedness} and any $x \in P$, by~\eqref{EQ:structuretheoremforapplication},
\[
    x + \lambda x^0 \in P \: (\lambda \geq 0) \quad \text{and} \quad f(x + \lambda x^0) = f(x_\mU) - \langle w, x + \lambda x^0 \rangle = f(x) - \lambda \to -\infty \quad (\lambda \to \infty).
\]
This shows condition \eqref{EQ:conditionforunboundedness} is sufficient.
To show it is also necessary, assume that $f$ is unbounded from below.
This means there is a sequence $(x^\ell)_{\ell \geq 1} \subseteq P$ of points with $\lim_{\ell \to \infty} f(x^\ell) = -\infty$.
Using a compactness argument, a subsequence of the normalized points $x^\ell/\|x^\ell\|$ converges to a unit vector~$x^0$.
If $\mU \oplus \spann{w} = \R^n$, using~\eqref{EQ:structuretheoremforapplication}, we are able to show that this $x^0$ satisfies \eqref{EQ:conditionforunboundedness}.
If not, we need to do a more careful limit argument to ensure that $x^0 \not\in (\mU \oplus \spann{w})^\perp$; see \Cref{SEC:fullproofconditionforunbounded}.

\paragraph{Bounding the norm in a subspace.}
From now on, we assume that $f$ is bounded from below on~$P$.
Recall that our goal is to show that $f$ attains its minimum at a point of small norm.
The structure theorem (cf. \eqref{EQ:structuretheoremforapplication}) shows that the value of $f(x)$ only depends on the component of $x$ in the subspace~${\mU \oplus \spann{w}}$.
Thus, we cannot bound the norm of all minimizers.
Instead, we first show that any minimizer has small norm in the subspace $\mU \oplus \spann{w}$, and then show we can lift to the full space without increasing the norm too much. 
Using the structure theorem, after writing~$q$ as its degree-2 Taylor expansion around $0$ and using $\mu$-strong convexity of $q$, we get
\begin{equation}\label{EQ:boundonf}
    f(x) \geq q(0) -  \|\nabla q(0)\| \cdot \|Ux\| + \frac{\mu}{2} \|Ux\|^2 - \langle w, x \rangle  \quad (\forall x \in \R^n).
\end{equation}
We want to show that if $x^*$ is a minimizer of $f$ on $P$, then both $\|Ux^*\|$ and $|\langle w, x^*\rangle|$ are small.
To do so, we will show that for any $x 
\in P$, we have that $|\langle w, x\rangle|$ is large only if $\|Ux\|$ is large as well. Then, since $\|Ux\|$ appears quadratically in~\eqref{EQ:boundonf}, we can conclude that both quantities must be small for minimizers.

Formally, since $f$ is bounded from below on $P$, there does not exist an $x^0$ as in~\eqref{EQ:conditionforunboundedness}.
Farkas' lemma gives us a witness for this fact; namely a vector $\lambda \geq 0$ and a vector $z \in \R^n$ such that $A^\top \lambda + U^\top z - w = 0$. We thus get $\langle w, x\rangle = \langle z, Ux \rangle + \langle \lambda, Ax \rangle$ for any $x \in \R^n$.
For $x \in P$, we have $Ax \leq b$.
Together with $\lambda \geq 0$, this implies that $\langle \lambda, A x \rangle \leq \lambda^\top b$ is uniformly bounded for all $x \in P$.
Thus, we have
\begin{equation}\label{EQ:Farkasidentity}
   |\langle w, x\rangle| \leq \|z\| \cdot \|Ux\| + \|\lambda\| \cdot \|b\| \quad (\forall x \in P).
\end{equation}
Using this in \eqref{EQ:boundonf} we get
\begin{equation}\label{EQ:boundonf2}
    f(x) \geq q(0) -  \|\nabla q(0)\| \cdot \|Ux\| - \|z\| \cdot \|Ux\| - \|\lambda\| \cdot \|b\| + \frac{\mu}{2} \|Ux\|^2 \quad (\forall x \in P).
\end{equation}
As the the right-hand side of \eqref{EQ:boundonf2} is a quadratic polynomial in $\|Ux\|$ with positive leading coefficient, $\|Ux^*\|$ needs to be small for any minimizer $x^*$ of $f$ on $P$.
Combining this with \eqref{EQ:Farkasidentity}, this implies that $|\langle w, x^* \rangle |$ is small.
To finish the argument, we need to relate these bounds to the bit length of the input.
Before we do so, we first show how to lift this bound to the full space.

\paragraph{Lifting to the full space.}\label{SEC:TO:norminVdirectionalsosmall}
So far, we have shown that for any minimizer $x^*$ of $f$ on~$P$, the norm in the subspace $\mU \oplus \spann{w}$ needs to be small.
The norm of $x^*$ in the full space might still be large and the projection to $\mU \oplus \spann{w}$ might not be feasible.
Thus, it remains to lift this projection to a feasible solution in the full space without increasing the norm too much.
This lifted point will automatically be a minimizer of $f$ on $P$ since adding any vector in $(\mU \oplus \spann{w})^\perp$ does not change the value of~$f$.
Formally, let $x^*$ be any minimizer of $f$ on $P$.
We want to find  
\begin{equation}\label{EQ:smallnorminVsystem}
    x' \in P: \:  x'_{\mU \oplus \spann{w}} = x^*_{\mU \oplus \spann{w}}
\end{equation}
with $x'$ having small norm.
The condition \eqref{EQ:smallnorminVsystem} defines a linear system.
It is feasible since $x^* \in P$.
If this system had small bit length, we could conclude immediately that it has a solution $x'$ of small bit length.
However, the right-hand side of \eqref{EQ:smallnorminVsystem} need not even be rational (we only know its norm is small).
We show that the matrix of this linear system has small bit length and the vector of the system has small norm, which we use to show there is a solution of small norm; see~\Cref{SEC:fullproofboundinVdirection}.
\paragraph{Bit complexity.}
It remains to relate the norm bounds on $x^*$ derived above to  the bit length of the input.
For this, we need to bound all the terms that occur in \eqref{EQ:Farkasidentity} and \eqref{EQ:boundonf2}.
By the structure theorem we have that $\bc(q) \leq \poly(\enc{f})$ and $\mu \geq 2^{-\poly(\enc{f})}$ and thus the terms involving $q$ and $\mu$ can be bounded appropriately.
Using a quantitative version of Farkas' lemma we also get that the terms involving $z$ and $\lambda$ can be bounded in terms of the bit length of the input.
This allows us to conclude that the norm of any minimizer $x^*$ in the subspace $\mU \oplus \spann{w}$ is at most $2^{\poly(\enc{f}, \, \enc{P})}$.
The lifting argument then also shows that there exists a minimizer $x^*$ that has norm at most $2^{\poly(\enc{f}, \, \enc{P})}$ in the full space, which completes the proof.

\section{Preliminaries}
\label{SEC:Preliminaries}

\subsection{Notations}

\paragraph{Linear algebra.} A symmetric matrix $M \in \R^{N \times N}$ is positive (semi)definite if all of its eigenvalues are greater than (or equal to) $0$.
For a positive semidefinite matrix, we denote its smallest eigenvalue by~$\lambda_{\min}(M)$, and its smallest \emph{nonzero} eigenvalue by~$\lambda_{\min}^+(M)$.
For a subspace $\mV \subseteq \R^N$, and a vector $x \in \R^N$, we write $x_\mV \in \mV$ for the projection of $x$ onto $\mV$.
For a subset $S \subseteq \R^N$, we define the affine hull $\aff(S)$ of $S$ as the set of all affine combinations of elements of $S$, i.e., of points $\sum_{i=1}^k \alpha_i x^i$ where $k \geq 1$, $x^i \in S$ for all $1 \leq i \leq k$ and $\sum_{i=1}^k \alpha_i = 1$. We call points $x^1, \ldots, x^k \in \R^N$ affinely independent if no point is an affine combination of the other points.

\paragraph{Smooth (convex) functions.}
For a sufficiently smooth $\varphi : \R^n \to \R$, 
we write $\nabla \varphi = (\partial \varphi / \partial x_i)_{1 \leq i \leq n}$ for its gradient, and $\nabla^2 \varphi = \big(\partial^2 \varphi / \partial x_i \partial x_j\big)_{1 \leq i, j \leq n}$ for its Hessian. For a vector $v \in \R^n$, we write ${\mynabla{v} \varphi(x) = \langle v, \nabla \varphi(x) \rangle}$ for the \emph{unnormalized} directional derivative. 
For $c \in \R$, we write $\varphi \equiv c$ if~$\varphi(x) = c$ for all $x \in \R^n$. 
We say $v \in \R^n$ is a \emph{direction of linearity} of $\varphi$ if $\mynabla{v} \varphi \equiv c$ for some $c \in \R$, i.e., the directional derivative is constant. 
For $\mu > 0$, we say that $\varphi$ is $\mu$-\emph{strongly convex} if $x \mapsto \varphi(x) - (\mu / 2) \cdot \|x\|^2$ is convex, or equivalently if $\lambda_{\min} \big(\nabla^2 \varphi(x)\big) \geq \mu$ for all $x \in \R^n$. If there exists any such $\mu > 0$, we say~$\varphi$ is strongly convex.

\paragraph{Polynomials.} We write $\R[x]$ (resp. $\Q[x]$) for the space of real (resp. rational) polynomials in~$n$ variables $x = (x_1, x_2, \ldots, x_n)$, with
(monomial) basis $\{ x^\alpha : \alpha \in \N^n\}$. For $p(x) = \sum_{\alpha} p_\alpha x^\alpha$, we write $\supp{p} \subseteq \N^n$ for the set of~$\alpha$ with $p_\alpha \neq 0$. 
We write $h_p(x) \coloneqq \det \big( \nabla^2 p(x) \big)$ for the Hessian determinant (or hessian) of $p$, which is itself a polynomial of degree at most $n \cdot \mathrm{deg}(p)$.

\paragraph{Bit length.} The bit length $\bc(k)$ of an integer $k$ is $\max \{1, \lceil \log_2 |k| \rceil \}$, which is the smallest natural number such that $|k| \leq 2^{\bc(k)}$. 
The bit length of a maximally simplified fraction $r = p/q \in \Q$ is $\bc(p) + \bc(q)$. The bit length of a rational vector or matrix is the sum of the bit lengths of its entries. 
The bit length of a rational polynomial $f \in \Q[x]$ is the sum of the bit lengths of its (nonzero) coefficients in the monomial basis. We write $\enc{f} = \Theta(n + 
\mathrm{deg}(f) + \bc(f))$ for the total encoding length of $f$. For a convex polyhedron $P = \{Ax \leq b\}$ given by a matrix $A \in \Q^{m \times n}$ and vector~${b \in \Q^m}$, we similarly write $\enc{P} = \bc(A) + \bc(b)$.
When $P = \R^n$, we assume that $\enc{P} \geq n$.

\subsection{The ellipsoid method}
We need the following two statements about the ellipsoid method.
\begin{proposition}[{\cite[Theorem 14.1]{Schrijver1994}}]\label{PROP:solvelinearprogram}
    Let $A \in \Q^{M \times N}$ and $b \in \Q^M$. Consider the linear program $P = \{Ax \leq b\}$.
    In polynomial time in $\enc{P}$ we can check whether the program is feasible and if so, compute a feasible solution.
    In particular, this feasible solution has bit size polynomial in $\enc{P}$.
\end{proposition}
\begin{restatable}{proposition}{PROPapplicationellipsoid}\label{PROP:applicationellipsoid}
    Let $R > 0$, $\varepsilon > 0$.
    Let $f \in \Q[x]$ be a polynomial.
    Let $P = \{Ax \leq b\}$ be a (nonempty) polyhedron.
    Then, there exists an algorithm that outputs $\Tilde{x} \in P$ such that $f(\Tilde{x}) \leq \min_{x \in P \cap B_R(0)} f(x) + \varepsilon$ in time $\poly(\enc{f},\, \enc{P},\, \log(R),\, \log(1/\varepsilon))$.
\end{restatable}

All the ideas needed to prove \Cref{PROP:applicationellipsoid} are standard results about the ellipsoid method.
Similar statements can for example be found in \cite{Schrijver1994, GrötschelLovaszSchrijver:ellipsoid, Vishnoi:convexoptimization}.
However, as we were unable to find the exact statement needed for our application (i.e., \Cref{PROP:applicationellipsoid}) we include a proof of this statement in \Cref{APP:Ellipsoid} for completeness.
Our proof mainly relies on the following standard results about the ellipsoid method:
\begin{itemize}
    \item It is known how to get \Cref{PROP:applicationellipsoid} if $P$ is full-dimensional and in a real model of computation where we are allowed to take square roots (see e.g. \cite[Chapter 13]{Vishnoi:convexoptimization}).
    \item The ellipsoid method can be implemented in the bit model (see e.g. \cite[Chapter 3]{GrötschelLovaszSchrijver:ellipsoid}).
    \item One can, again using the ellipsoid method, find the affine hull of $P$, which allows us to reduce to the full-dimensional case (see e.g. \cite[Chapters 5 and 6]{GrötschelLovaszSchrijver:ellipsoid}).
\end{itemize}
In \Cref{APP:Ellipsoid} we combine these ideas and give a complete proof of \Cref{PROP:applicationellipsoid} (in the bit model).
The main idea is to first find the affine hull of $P$, 
then project to this affine hull (which does increase the bit complexity by at most a polynomial factor) and finally apply the ellipsoid method on the projected polyhedron, which is now full-dimensional.

\section{Structure theorem for convex polynomials}
\label{SEC:ProofStructureThm}
In this section, we give the full proof of our structure theorem for convex polynomials, which we restate for convenience.

\quantstructurethm*\noindent
Our proof relies on two key technical tools. The first provides a $\mu$-strongly convex lower bound on a convex polynomial $f \in \Q[x]$ whose hessian $h_f(x) \coloneqq \det \left(\nabla^2 f(x)\right)$ is not identically zero.
\begin{restatable}{proposition}{quantstronglyconvex}
    \label{PROP:quantstronglyconvex}
    Let $f \in \Q[x]$ be a convex polynomial with $h_f \not \equiv 0$.
    Then, there is a $\mu$-strongly convex, quadratic $q \in \Q[x]$ with $\bc(q) \leq \poly(\enc{f})$ such that $f(x) \geq q(x)$ for all $x \in \R^n$. Moreover,
    $
        \mu \geq 2^{-\poly(\enc{f})}.
    $
\end{restatable} \noindent
The second shows that any convex polynomial $f \in \Q[x]$ can be written as the sum of a linear function and a polynomial $\hat f$ (of fewer variables) with nonzero hessian $h_{\hat f} \not \equiv 0$.
\begin{proposition} \label{PROP:quantdecomposition}
Let $f \in \Q[x]$ be a convex polynomial. Then, there is a subspace $\mU \subseteq \R^n$, a matrix $U \in \Q^{k \times n}$ whose orthogonal rows span $\mU$, and a rational vector $w \in \mU^\perp$ (possibly $w = 0$), such that
\begin{equation} \label{EQ:quantdecomp}
    f(x) = \hat f (Ux) - \langle w, x\rangle \quad (\forall \, x\in \R^n),
\end{equation}
where the polynomial $\hat f : \R^k \to \R$ defined by $\hat f(Ux) = f(x_\mU)$ satisfies $h_{\hat f} \not \equiv 0$. Moreover, we can compute $U, w$ and~$\hat f$ in polynomial time in $\enc{f}$. In particular, $U, w, \hat f$ have bit length $\poly(\enc{f})$.
\end{proposition}
Together, these statements immediately imply~\Cref{THM:quantstructure}. Indeed, one first decomposes $f$ according to~\eqref{EQ:quantdecomp}, and then applies~\Cref{PROP:quantstronglyconvex} to $\hat f$.  In the remainder of this section, we prove~\Cref{PROP:quantstronglyconvex} and~\Cref{PROP:quantdecomposition} separately.

\subsection{\texorpdfstring{A strongly convex lower bound: Proof of~\Cref{PROP:quantstronglyconvex}}{A strongly convex lower bound}}
Let $f \in \Q[x]$ be a convex polynomial of degree $d$ with $h_f \not \equiv 0$. For any $a \in \R^n$, we know from~\cite[Lemma 5]{Ahmadi2024} that
\[
    f(\x) \geq f(a) + \langle \nabla f (a), 
        \, \x-a \rangle + \frac{\langle x-a, \, \nabla^2 f(a) \cdot (x-a) \rangle}{4d^2} \quad (\forall x \in \R^n).
\]
If $\nabla^2 f(a) \succ 0$, this gives us a $\mu$-strongly convex, quadratic lower bound on $f$, namely,
\[
f(\x) \geq q(x) \coloneqq f(a) + \langle \nabla f (a), \x-a \rangle + \mu \cdot \|x-a\|^2  \quad (\forall 
x \in \R^n),
\]
where
\begin{equation*}
    \mu = \mu(a) = \frac{\lambda_{\min}(\nabla^2 f (a))}{2d^2} > 0.
\end{equation*}
Note that, if $a \in \Q^n$, then $\bc(q) \leq \poly(\enc{f},\, \bc(a))$. Thus, it suffices to find an $a \in \Q^n$ with $\bc(a) \leq \poly(\enc{f})$ such that $\mu(a) \geq 2^{-\poly(\enc{f})}$. To do so, we establish the following two facts:
\begin{fact} \label{Fact:nonzero}
    If $h_f \not \equiv 0$, then there
    there is an $a \in \Q^n$ with $\bc(a) \leq \poly(n, d)$ such that $h_f(a) > 0$. 
\end{fact}
\begin{fact} \label{Fact:bound}
    For any $a \in \Q^n$ with $\bc(a) \leq \poly(\enc{f})$ and $h_f(a) > 0$, we have $\mu(a) \geq 2^{-\poly(\enc{f})}$.
\end{fact}

\subsubsection{Polynomial identity testing: proof of~\Cref{Fact:nonzero}}
\label{SEC:PROOF:smallnonzero}
The Schwartz-Zippel lemma is a standard tool in (probabilistic) polynomial identity testing.
It states that, for any polynomial $p \not\equiv 0$, and $S \subseteq \R$ finite,
\[
    \mathbb{P}_{x \, \sim \, \mathrm{Unif}(S^n)}\big[p(x) = 0\big] \leq \frac{\mathrm{deg}(p)}{|S|},
\]
where $\mathrm{Unif}(S^n)$ is the uniform distribution on $S^n \subseteq \R^n$. The following is a deterministic version:
\begin{lemma}[Schwartz-Zippel, see~{\cite[Corollary 1]{Schwartz1980}}] \label{LEM:ZS} Let $p \in \R[\x]$ be a nonzero polynomial of degree~$d$. Let $S = \{0, 1, \ldots, d\}$. Then, there is an $a \in S^n$ with $p(a) \neq 0$. 
\end{lemma}
Applying the deterministic Schwartz-Zippel lemma to the Hessian determinant $h_f \not \equiv 0$ yields the desired point $a \in \Q^n$ of small bit size for which $h_f(a) > 0$ (recall that $h_f(x) \geq 0$ for all $x \in \R^n$).
\begin{corollary} \label{COR:ZSHessian}
    Let $p$ be a polynomial in $n$ variables of degree~$d$ and assume that $h_p \not \equiv 0$. Then, there is an $a \in \Z^n$ with $\bc(a) \leq O(n \log (dn))$ such that $h_p(a) = \det \left(\nabla^2 p (a) \right) \neq 0$.
\end{corollary}
\begin{proof}
    By assumption, the Hessian determinant $h_p$ of $p$ is a nonzero polynomial in $n$ variables of degree at most $dn$. By \Cref{LEM:ZS}, there is thus an $a \in \{0, 1, \ldots, dn\}^n$ with $h_p(a) \neq 0$. This $a$ has bit size at most $O(n \log (dn))$.
\end{proof}

\subsubsection{Spectral bounds for rational matrices: proof of~\Cref{Fact:bound}} \label{SEC:PROOF:rationalmatrix}
Now, let $a \in \Q^n$ with $\bc(a) \leq \enc{f}$ and $h_f(a) > 0$. We aim to show that $\mu(a) \geq 2^{-\poly(\enc{f})}$, for which it suffices to show that $\lambda_{\min}(\nabla^2 f (a)) \geq 2^{-\poly(\enc{f})}$. 
For this, note that $\nabla^2 f (a) \succ 0$ is a rational matrix of bit size at most $\poly(\enc{f})$. Indeed, the entries of $\nabla^2 f(a)$ each have bit length at most $\poly(\enc{f}, \, \bc(a))$. The eigenvalues of such matrices can be bounded away from $0$ as follows.

\begin{lemma}[{\cite[Lemma 3.1]{Raghavendra2017}}] \label{LEM:integermatrix}
Let $M \in \Z^{N \times N}$ be a positive semidefinite matrix with ${|M_{ij}| \leq B}$ for all $1 \leq i, j \leq N$. Then $\lambda_{\min}^+(M) \geq (BN)^{-N}$.
\end{lemma}

\begin{corollary} \label{COR:rationalmatrix}
Let $M \in \Q^{N \times N}$ be a positive definite matrix. Then,
$
\lambda_{\min}(M) \geq 2^{-\poly(\bc(M))}.
$
\end{corollary}
\begin{proof}
    Let $C \in \Z$ be the least common multiple of the denominators of the entries of $M$. 
    Note that $\bc(C) \leq \poly(\bc(M))$, and so $C \cdot M$ is an integer matrix whose entries are bounded in magnitude by $2^{\poly(\bc(M))}$. Applying~\Cref{LEM:integermatrix} yields
    \[
        \lambda_{\min}(M) = \frac{1}{C} \cdot \lambda_{\min}(C \cdot M) \geq \frac{1}{2^{\poly(\bc(M))}} \cdot(2^{\poly(\bc(M))} \cdot N)^{-N} \geq 2^{-\poly(\bc(M))}. \qedhere
    \]
\end{proof}

\subsection{\texorpdfstring{A Hessian decomposition theorem: Proof of~\Cref{PROP:quantdecomposition}}{A Hessian decomposition theorem}}
The primary tool in our proof of~\Cref{PROP:quantdecomposition} is the following characterization of convex functions whose Hessian determinant is identically zero.
\begin{restatable}{proposition}{Hessianvanishes}
    \label{PROP:Hessianvanishes}
    Let $\varphi: \R^n \to \R$ be twice continuously differentiable and convex. Then, $\varphi$ has a direction of linearity if and only if its Hessian is nowhere definite:
    \[
        \det \left(\nabla^2 \varphi \right) \equiv 0 \iff \exists v \in \R^n \setminus \{ 0 \}, \, c \in \R : \mynabla{v} \varphi \equiv c.
    \]
\end{restatable}
To establish this proposition, we consider $\nabla \varphi$ as function $\R^n \to \R^n$, and  investigate its range 
\[
    \nabla \varphi(\R^n) \coloneqq \{ \nabla\varphi(x) : x \in \R^n \} \subseteq \R^n.
\]
We aim to show that $\det \left(\nabla^2{\varphi}\right) \equiv 0$ if and only if $\nabla \varphi(\R^n)$ is contained in an affine hyperplane. First, we recall that that the range of the gradient of a convex function is \emph{almost convex}.
\begin{lemma}[cf. {\cite[Corollary 2]{Minty1964}}] \label{LEM:convexgradient}
    Let $\varphi : \R^n \to \R$ be  convex and continuously differentiable. Then the range of the gradient of $\varphi$ is \emph{almost convex}, i.e., we have $\mathrm{int}(\mathrm{conv}(\nabla \varphi (\R^n))) \subseteq \nabla \varphi (\R^n) $.
\end{lemma}
\begin{proof}
    The \emph{convex conjugate} of $\varphi : \R^n \to \R$ is the map 
    \[
        \varphi^* : \R^n \to \R \cup \{\infty\}, \quad v \mapsto \sup_{x \in \R^n} \left\{ \langle v, x \rangle - \varphi(x) \right\}.
    \]
    The function $\varphi^*$ is itself convex (since it is the supremum of affine functions). Furthermore, $\varphi^* \not\equiv \infty$ (i.e., $\varphi^*$ is a \emph{proper} convex function) \cite[Theorem 12.2]{Rockafellar1970}.
    Lastly, for any $x, v \in \mathbb{R}^n$, $v$ is a subgradient of $\varphi$ at $x$ if and only if $x$ is a subgradient of $\varphi^*$ at $v$ \cite[Theorem 23.5]{Rockafellar1970}. 
    
    Now, since $\varphi$ is differentiable, we get
    \begin{equation} \label{EQ:subgradient}
    \nabla \varphi (\R^n) = \{v \in \R^n : \varphi^* \text{ has a subgradient at } v\}.
    \end{equation}
    We give a convex set $C$ such that $\mathrm{int}(C) \subseteq \nabla \varphi (\R^n) \subseteq C$.
    This immediately implies that $\nabla \varphi (\R^n)$ is almost convex, since we then have $\mathrm{conv}(\nabla \varphi (\R^n)) \subseteq C$, and thus also $\mathrm{int}(\mathrm{conv}(\nabla \varphi (\R^n))) \subseteq \mathrm{int}(C)$.
    
    Let $C \coloneqq \{v \in \R^n \mid \varphi^*(v) \neq \infty\}$ be the \emph{effective domain} of $\varphi^*$.
    Since $\varphi^* \not\equiv \infty$, we find that~$\varphi^*$ has no subgradient at any $v \in \R^n$ with $\varphi^*(v) = \infty$. 
    Hence, using~\eqref{EQ:subgradient}, we get $\nabla \varphi (\R^n) \subseteq C$.
    On the other hand, because $\varphi^*$ is a proper convex function, it has a subgradient at any point in the relative interior of its effective domain~\cite[Theorem 23.4]{Rockafellar1970}.
    Thus, $\mathrm{int}(C) \subseteq \mathrm{relint}(C) \subseteq \nabla \varphi(\R^n)$, which completes the proof.
\end{proof}
Then, we observe that $\nabla^2 \varphi$ is the Jacobian of $\nabla \varphi$. Thus, if $\det \left(\nabla^2 \varphi\right) \equiv 0$, we find that $\nabla \varphi(\R^n)$ has Lebesgue measure zero by Sard's theorem.
\begin{theorem}[Sard's theorem]  \label{THM:Sard} 
    Let $\psi : \R^n \to \R^n$ be continuously differentiable. Let 
    \[
        X \coloneqq \{ x \in \R^n : \det \left(\nabla \psi(x) \right) = 0\}.
    \]
    Then $\psi(X) \subseteq \R^n$ has Lebesgue measure zero. 
\end{theorem}
\begin{proof}[Proof of~\Cref{PROP:Hessianvanishes}]
    First, assume that $\mynabla{v} \varphi \equiv c$ for some $v \in \R^n \setminus \{0\}$, $c \in \R$. Then, ${\mynabla{w} \mynabla{v} \varphi(x) = 0}$ for all $w \in \R^n$ and $x \in \R^n$, meaning $\nabla^2 \varphi(x) \cdot v = 0$. In particular, $\det \left( \nabla^2 \varphi(x) \right) = 0$.

    Now, assume $\det \left(\nabla^2 \varphi\right) \equiv 0$. By the above, this implies that $\nabla \varphi(\R^n)$ is an almost convex set of Lebesgue measure $0$. But that means $\nabla \varphi(\R^n)$ is contained in an affine hyperplane, for, if not, the interior of its convex hull would contain (the interior of) a simplex of positive measure. Thus, there is a $v \in \R^n \setminus \{0\}$, and $c \in \R$ such that $\langle y, v \rangle = c$ for all $y \in \nabla\varphi(\R^n)$. In other words, $\mynabla{v} \varphi \equiv c$.
\end{proof}

\begin{proof}[Proof of~\Cref{PROP:quantdecomposition}]
With \Cref{PROP:Hessianvanishes} in hand, we are able to identify the subspace $\mU \subseteq \R^n$ and vector $w \in \Q^n$ that feature in~\Cref{PROP:quantdecomposition}. Let $f \in \Q[x]$ be a convex polynomial. Consider the joint support 
\[
    S = \{ 0 \} \cup \bigcup_{i=1}^n \, \supp{\partial f / \partial x_i} \subseteq \N^n
\]
of the partial derivatives of $f$, which is of size $|S| \leq n \cdot |\supp{f}| + 1$. (Here, we have added $0 \in \N^n$ for technical reasons that become clear shortly.) Set $\mathcal{R} = \spann {x^\alpha : \alpha \in S} \subseteq \R[x]$. We view the gradient of $f$ as a linear operator $\R^n \to \mathcal{R}$, namely
\[
    \nabla f : v \mapsto \mynabla{v} f = \sum_{i=1}^n v_i \cdot \frac{\partial f}{\partial x_i} \in \mathcal{R} \quad (v \in \R^n).
\]
Importantly, this operator can be expressed as a rational matrix $\mat{\nabla f}$ of size $|S| \times n$, where the $i$th column represents the expansion of $\partial f / \partial x_i$ in the (partial) monomial basis. Moreover, this matrix has bit size at most $\poly(\enc{f})$.  
Write~${\mathds{1} \in \Q^{|S|}}$ for the expansion of the constant polynomial $x \mapsto 1$ (which exists as $0 \in S$). In light of~\Cref{PROP:Hessianvanishes}, we consider the subspace
\[
    \mL \coloneqq \{ v \in \R^n : \mynabla{v} f \equiv c \text{ for some } c \in \R\} =  \{ v \in \R^n : \mat{\nabla f} \cdot v \in \spann{\mathds{1}}\}.
\]
Since $\mL$ is the inverse image of a subspace of dimension $1$, we can write 
\[
    \mL = \ker \left(\mat{\nabla f}\right) \oplus \spann{w},
\]
where $w \in \ker \left(\mat{\nabla f}\right)^\perp$ is either zero, or satisfies $\mat{\nabla f} \cdot w = -\mathds{1}$. Then, we set $\mU = \mL^\perp$, and write $\mV = \ker \left(\mat{\nabla f}\right)$, $\mW = \spann{w}$. By construction, we have
\[
    f(x) = f(x_{\mU} + x_{\mV} + x_{\mW}) = f(x_\mU) - \langle w, x \rangle.
\]
We can compute $w$, and an orthogonal basis $U_1, \ldots, U_k$ for $\mU$ in polynomial time in $\enc{f}$. In particular, $w$ and the matrix $U \in \Q^{k \times n}$ whose rows are $U_1, \ldots, U_k$ have polynomial bit size in~$\enc{f}$. This follows from standard (bit-)complexity results in linear algebra, see \Cref{APP:LA}.

For $1 \leq i \leq k$, write $\overline{U}_i = U_i / \|U_i\|^2$, and note that $\bc(\overline{U}_i) \leq \poly(\enc{f})$. Consider the $k$-variate polynomial $\hat f \in \Q[y_1, \ldots, y_k]$ defined by $\hat f(y) = f(\sum_{i=1}^k y_i \cdot \overline{U}_i)$.  Note that $\hat f$ is convex, and that $\enc{\hat f} \leq \poly(\enc{f})$. Furthermore, for any $x \in \R^n$,
\[
    \hat f (Ux) = f \left(\sum_{i=1}^k \frac{\langle x, U_i \rangle}{\, \|U_i\|^2} \cdot U_i \right) = f(x_\mU).
\]
It remains to show that $h_{\hat f} \not \equiv 0$. By~\Cref{PROP:Hessianvanishes}, this is equivalent to showing it has no directions of linearity. Suppose that it did, i.e., that there is a nonzero $\hat v \in \R^k$ such that $\mynabla{\hat v} \hat f \equiv c$ for some~${c \in \R}$.  Then, writing $e_j \in \R^n$ for the $j$-th standard basis vector,
\[
    \mynabla{e_j} \hat f (y) = \mynabla{\overline{U}_j} f\left(\sum_{i=1}^k y_i \cdot \overline{U}_i \right).
\]
Now, write $u = \sum_{i=1}^k \hat v_i \overline{U}_i \in \mU = \mL^\perp$. Note that $u \neq 0$.
Then, as $\langle w, u \rangle = 0$, for any $x \in \R^n$,
\[
\mynabla{u} f(x) = \mynabla{u} f \left (\sum_{i=1}^k \langle x, U_i \rangle \cdot \overline{U}_i \right) = \mynabla{\hat v} \hat f(Ux) = c,
\]
a contradiction.
\end{proof}

\section{Applications in convex polynomial programming}
\label{SEC:ProofApplications}
In this section, we give a full proof of \Cref{THM:mainpolyhedron} and \Cref{COR:mainpolyhedron}, which we restate below.

\THMmainpolyhedron*

\CORmainpolyhedron*

For the proof of \Cref{THM:mainpolyhedron} and \Cref{COR:mainpolyhedron} we need \Cref{THM:quantstructure}, which we also restate below for convenience. 

\quantstructurethm*

In particular, this gives us the following lower bound on $f$:
\begin{equation}\label{EQ:fullprooflowerboundonf}
    f(x) \geq q(0) -  \|\nabla q(0)\| \cdot \|Ux\| - \langle w, x \rangle + \frac{\mu}{2} \|Ux\|^2.
\end{equation}

We define $\mW \coloneqq \spann{w}$ for $w$ as in \Cref{THM:quantstructure}.
In order to prove \Cref{THM:mainpolyhedron} and \Cref{COR:mainpolyhedron}, we also need the following three lemmas.
\Cref{LEM:fullproofconditionforunbounded} shows how we can efficiently detect whether $f$ is unbounded from below.
\Cref{LEM:fullproofboundinUandWdirection} then shows that any minimizer needs to have small norm in the subspace $\mU \oplus \mW$.
Finally, \Cref{LEM:fullproofboundinVdirection} shows how to lift this norm bound in a subspace to the full space.

\begin{lemma}\label{LEM:fullproofconditionforunbounded}
    Let $f \in \Q[x]$ be a convex polynomial and let $P = \{Ax \leq b\}$ be a (nonempty) polyhedron.
    Let $\mU$ and $w$ as in \Cref{THM:quantstructure}.
    Then, $f$ is unbounded from below on $P$, if and only if there exists $x^0 \in \R^n$ with $Ax^0 \leq 0$, $x^0 \in \mU^\perp$ (or equivalently $Ux^0 = 0$ for $U$ as in \Cref{THM:quantstructure}) and $\langle w, x^0 \rangle = 1$.
\end{lemma}

\begin{lemma}\label{LEM:fullproofboundinUandWdirection}
    Let $f \in \Q[x]$ be a convex polynomial and let $P = \{Ax \leq b\}$ be a (nonempty) polyhedron.
    Let $\mU$ and $w$ as in \Cref{THM:quantstructure} and let $x \in P$.
    Assume that $x$ satisfies $\|x_{\mU \oplus \mW}\| > 2^{\poly(\enc{f}, \,\enc{P})}$.
    Then there is a point $x' \in P$ with $\|x'\| \leq 2^{\poly(\enc{f}, \,\enc{P})}$  and $f(x') < f(x)$.
    In particular, if $f$ is bounded from below on $P$, then every minimizer $x^* \in P$ of $f$ on $P$ satisfies $\|x^*_{\mU \oplus \mW}\| \leq 2^{\poly(\enc{f}, \,\enc{P})}$.
\end{lemma}

\begin{lemma}\label{LEM:fullproofboundinVdirection}
    Let $P = \{Ax \leq b\}$ be a (nonempty) polyhedron. Let $\mZ$ be a linear subspace of $\R^n$ spanned by orthogonal vectors $x^1, \ldots, x^k$. Let $x^* \in P$. There is an $x' \in P$ such that $x^*_\mZ = x'_\mZ$ and
    \[
        \|x'\| \leq \poly\left(\|x^*_{\mZ}\|,\,2^{\poly(\enc{P}, \,\bc(\mZ))}\right),
    \]
    where $\bc(\mZ) = \bc(x^1) + \ldots + \bc(x^k)$.
\end{lemma}

We now first give a formal proof of \Cref{THM:mainpolyhedron} and \Cref{COR:mainpolyhedron} given the lemmas and then prove \Cref{LEM:fullproofconditionforunbounded,LEM:fullproofboundinUandWdirection,LEM:fullproofboundinVdirection} in \Cref{SEC:fullproofconditionforunbounded,SEC:fullproofboundinUandWdirection,SEC:fullproofboundinVdirection}.

\begin{proof}[Proof of \Cref{THM:mainpolyhedron}]
    Assume that $f$ is bounded from below (otherwise we are immediately done).
    We claim that
    \begin{equation}\label{EQ:fullproofPattainsminimumonball}
        \inf_{x \in P} f(x) = \inf_{x \in P :\: \|x\| \leq R} f(x)
    \end{equation}
    for some $R = 2^{\poly(\enc{f}, \,\enc{P})}$.
    Since the set on the right-hand side is compact and $f$ is continuous, this will allow us to conclude that $f$ attains its minimum.
    Moreover, it attains it at a point $x^*$ with $\|x^*\| \leq R$, which will complete the proof.
    Thus, it remains to prove \eqref{EQ:fullproofPattainsminimumonball}.
    Assume by contradiction that this is not true.
    Then, there is $x^0 \in P$ with
    \[
        f(x^0) < \inf_{x \in P :\: \|x\| \leq R} f(x).
    \]
    By \Cref{LEM:fullproofboundinUandWdirection}, we can assume without loss of generality that $\|(x^0)_{\mU \oplus \mW}\| \leq 2^{\poly(\enc{f}, \,\enc{P})}$.
    Otherwise, we could replace $x^0$ by a point ${x^0}'$ that satisfies this and has $f({x^0}') < f(x^0)$.
    By \Cref{LEM:fullproofboundinVdirection} (setting $\mZ = \mU \oplus \mW$ and $x^1, \ldots, x^k$ being the rows of $U$ and $w$), we can now find an $x' \in P$ such that $x^0_{\mU \oplus \mW} = x'_{\mU \oplus \mW}$ and
    \begin{equation}\label{EQ:fullproofpointofsmallnormforminimum}
        \|x' \| \leq \poly\left(\|x^0_{\mU \oplus \mW}\|,\, 2^{\poly(\enc{P}, \,\bc(\mU \oplus \mW))}\right) \leq 2^{\poly(\enc{f}, \,\enc{P})}.
    \end{equation}
    Here we used that by \Cref{THM:quantstructure} we have $\bc(\mU \oplus \mW) \leq \poly(\enc{f})$.
    Choosing $R$ as this upper bound, we have
    \[
        f(x') = f(x'_{\mU \oplus \mW}) = f(x^0_{\mU \oplus \mW}) = f(x^0) < \inf_{x \in P :\: \|x\| \leq R} f(x),
    \]
    which contradicts \eqref{EQ:fullproofpointofsmallnormforminimum}.
    Thus, \eqref{EQ:fullproofPattainsminimumonball} holds, which completes the proof.
\end{proof}

\begin{algorithm}
\caption{Algorithm for convex polynomial programming}\label{ALG:mainalgorithm}
\begin{algorithmic}[1]
\Input A polynomial $f$, a matrix $A \in \Q^{m \times n}$, a vector $b \in \Q^m$ (defining a nonempty polyhedron $P=\{Ax \leq b\}$), an error parameter $\varepsilon$
\Output Either $-\infty$ if $f$ is unbounded on $P$ or a point $\tilde{x} \in P$ with $f(\tilde{x}) \leq \min_{x \in \R^n} f(x) + \varepsilon$
\State Compute $U$ and $w$ as in \Cref{THM:quantstructure}.\label{ALGLINE:structuretheorem}
\If{there exists $x^0 \in \R^n$ with $Ax^0 \leq 0$, $Ux^0 = 0$ and $\langle w, x^0 \rangle = 1$}\label{ALGLINE:checkforunbounded}
\Return $-\infty$
\Else
\State \parbox[t]{0.928\linewidth}{Run the ellipsoid method from \Cref{PROP:applicationellipsoid} with input $f$, $P$, $R = 2^{\poly(\enc{f}, \,\enc{P})}$ (as in \Cref{THM:mainpolyhedron}) and $\varepsilon$ and \Return the output of the ellipsoid method.\label{ALGLINE:ellipsoid}}
\EndIf
\end{algorithmic}
\end{algorithm}
\begin{proof}[Proof of \Cref{COR:mainpolyhedron}]
Consider \Cref{ALG:mainalgorithm}, which we want to use to show \Cref{COR:mainpolyhedron}.
Correctness follows from \Cref{LEM:fullproofconditionforunbounded} and \Cref{THM:mainpolyhedron}.
The runtime is $\poly(\enc{f}, \,\enc{P},\, \log(1/\varepsilon))$.
Indeed, by \Cref{THM:quantstructure}, line~\ref{ALGLINE:structuretheorem} can be done in time $\poly(\enc{f})$.
The check in line~\ref{ALGLINE:checkforunbounded} can be done in time $\poly(\enc{f}, \,\enc{P})$ by \Cref{PROP:solvelinearprogram} since it checks for feasibility of a linear program and the bit lengths of $A$, $U$ and $w$ are bounded by $\poly(\enc{f}, \,\enc{P})$.
Finally, the ellipsoid method from \Cref{PROP:applicationellipsoid} in line~\ref{ALGLINE:ellipsoid} runs in time $\poly(\enc{f},\, \enc{P},\, \log(R),\, \log(1/\varepsilon))$, which is also $\poly(\enc{f}, \,\enc{P},\, \log(1/\varepsilon))$ since we have $R = 2^{\poly(\enc{f}, \,\enc{P})}$.
\end{proof}

\subsection{\texorpdfstring{Deciding unboundedness: Proof of \Cref{LEM:fullproofconditionforunbounded}}{Deciding unboundedness}}\label{SEC:fullproofconditionforunbounded}
In this section, we prove \Cref{LEM:fullproofconditionforunbounded}, i.e. we show that $f$ is unbounded on $P$ if and only if there exists an $x^0 \in \R^n$ with $Ax^0 \leq 0$, $x^0 \in \mU^\perp$ and $\langle w, x^0 \rangle = 1$.
We first show that this condition is sufficient.

\begin{proof}[Proof of \Cref{LEM:fullproofconditionforunbounded} (part 1: condition is sufficient)]
    Let $x \in P$ and consider the points $x + \lambda x^0$ for $\lambda \geq 0$. We have
    \[
        A(x + \lambda x^0) = Ax + \lambda Ax^0 \leq b + \lambda \cdot 0 = b
    \]
    and thus $x + \lambda x^0 \in P$ for any $\lambda \geq 0$. Furthermore
    \[
        f(x + \lambda x^0) = f((x + \lambda x^0)_\mU) - \langle w, x + \lambda x^0 \rangle.
    \]
    Since $x^0 \in \mU^\perp$, we have $(x + \lambda x^0)_\mU = x_\mU$. Thus, using $\langle w, x^0 \rangle = 1$,
    \[
        f(x + \lambda x^0) = f(x_\mU) - \langle w,x \rangle - \lambda,
    \]
    showing that $f(x + \lambda x^0) \to -\infty$ as $\lambda \to \infty$.
    Hence, if the condition is satisfied, then $f$ is unbounded from below on $P$.
\end{proof}

To show that the condition is also necessary, assume that $f$ is unbounded from below on~$P$.
Consider a sequence $(x^k)_{k \geq 1} \subseteq P$ with $\lim_{k \to \infty} f(x^k) = -\infty$.
Since the unit sphere $S^{n-1}$ is compact, a subsequence of the normalized points $x^k/\|x^k\|$ converges to a point $x^0$.
Unfortunately, $x^0$ could be in $(\mU \oplus \mW)^\perp$ (in particular meaning that $\langle w, x^0 \rangle = 0$).
Instead, we get the following claim about the projected normalized points.
(Note that, since $f(x) = f(x_{\mU \oplus \mW})$ for all $x \in \R^n$, the statements ${\lim_{k \to \infty} f(x^k) = -\infty}$ and $\lim_{k \to \infty} f(x^k_{\mU \oplus \mW}) = -\infty$ are equivalent.)
\begin{claim}\label{CL:limitinUplusW}
    Consider a sequence $(x^k)_{k \geq 1}$ such that $\lim_{k \to \infty} f(x^k) = -\infty$. Then a subsequence of the normalized projected points $x^k_{\mU \oplus \mW}/\|x^k_{\mU \oplus \mW}\|$ converges to a point $y^0 \in \mW$ with $\langle w, y^0 \rangle > 0$.
\end{claim}

Before we prove this claim, we show how to apply it to prove that the condition of \Cref{LEM:fullproofconditionforunbounded} is also necessary.
For this, we want to argue that we can lift this limit to a point in the full space satisfying the conditions of \Cref{LEM:fullproofconditionforunbounded}.
In order to do so, we need the following two propositions.
Recall that the recession cone $\recc(C)$ of a convex set $C$ is the set of directions in which $C$ extends to infinity, i.e.
\[
    \recc(C) = \{y \in \R^n : x+\lambda y \in C \:\:\forall x \in C\: \forall \lambda \geq 0\}.
\]
\begin{proposition}[{\cite[(Special case of) Corollary 8.3.4]{Rockafellar1970}}]\label{PROP:projectionreccesioncone}
    Let $C$ be a (nonempty) closed convex set.
    Let $\mZ$ be a subspace.
    Consider the projection
    \[
        C_\mZ = \{x \in \mZ : \exists x' \in \mZ^\perp \text{ with } x + x' \in C\}.
    \]
    Then the projection of the recession cone of $C$ is the recession cone of the projection $C_\mZ$.
\end{proposition}
\begin{proposition}[{\cite[Theorem 19.3]{Rockafellar1970}}]\label{PROP:projectionofpolyhedronispolyhedron}
    The projection of the polyhedron $P = \{Ax \leq b\}$ to a subspace~$\mZ$ is again a polyhedron, i.e., it can be described as
    \[
        P_{\mZ} = \{x \in \mZ : \exists x' \in (\mZ)^\perp \text{ with } x + x' \in P\} = \{\hat{A}x \leq \hat{b}\}.
    \]
\end{proposition}
Note that the dimensions of $\hat{A}$ and $\hat{b}$ can be exponential in $n$, but since we are only interested in proving existence of an $x^0$ as in \Cref{LEM:fullproofconditionforunbounded} this is not a problem.
We now show how to use \Cref{CL:limitinUplusW} and \Cref{PROP:projectionreccesioncone,PROP:projectionofpolyhedronispolyhedron} to prove \Cref{LEM:fullproofconditionforunbounded}.

\begin{proof}[Proof of \Cref{LEM:fullproofconditionforunbounded} (part 2: condition is necessary)]
    Since $f$ is unbounded, there is a sequence $(x^k)_{k \geq 1}$ of points in $P$ such that $\lim_{k \to \infty} f(x^k) = -\infty$.
    Thus, using \Cref{CL:limitinUplusW} we get a $y^0 \in \mW$ with 
    \[
        y^0 = \lim_{k \to \infty} \frac{x^k_{\mU \oplus \mW}}{\|x^k_{\mU \oplus \mW}\|}.
    \]
    Note that we can assume without loss of generality that $y^0$ is the limit of all the points (and not just of a subsequence).
    Since $y^0 \in \mW$, we in particular have $y^0 \in \mU^\perp$.

    We want to show that we can lift $y^0$ to a point $x^0$ (meaning to an $x^0$ such that $x^0 - y^0 \in (\mU \oplus \mW)^\perp$) with $Ax^0 \leq 0$.
    We have $x^k_{\mU \oplus \mW} \in P_{\mU \oplus \mW}$.
    By \Cref{PROP:projectionofpolyhedronispolyhedron}, $P_{\mU \oplus \mW} = \{\hat{A}x \leq \hat{b}\}$ for some $\hat{A}$ and~$\hat{b}$.
    Since we also have ${\lim_{k \to \infty} f(x_{\mU \oplus \mW}^k) = -\infty}$, it follows that $\lim_{k \to \infty} \|x_{\mU \oplus \mW}^k\| = \infty$ ($f$ is bounded over any compact set) and thus also
    \[
        \hat{A}y^0 = \lim_{k \to \infty}\frac{1}{\|x_{\mU \oplus \mW}^k\|}\hat{A}x_{\mU \oplus \mW}^k \leq \lim_{k \to \infty} \frac{1}{\|x_{\mU \oplus \mW}^k\|}\hat{b} = 0.
    \]
    Thus, $y^0$ is in the recession cone of $P_{\mU \oplus \mW}$, which by \Cref{PROP:projectionreccesioncone} is the projection of the recession cone $\recc(P) = \{x \in \R^n : Ax \leq 0\}$ of $P$.
    This implies that we can lift $y^0$ to a point $x^0 \in \recc(P)$.
    
    We then have $Ax^0 \leq 0$ by definition of $\recc(P)$.
    Since $x^0$ and $y^0$ only differ in $(\mU \oplus \mW)^\perp$, we still have $x^0 \in \mU^\perp$.
    Since $\langle w, x^0 \rangle = \langle w, y^0 \rangle > 0$, by rescaling, we can get $\langle w, x^0\rangle = 1$, which completes the proof.
\end{proof}

It remains to prove \Cref{CL:limitinUplusW}.

\begin{proof}[Proof of \Cref{CL:limitinUplusW}]
    Consider the projections $x^k_{\mU \oplus \mW}$.
    We have ${\lim_{k \to \infty} f(x^k_{\mU \oplus \mW}) = -\infty}$.
    Thus, by compactness of the unit sphere, a subsequence of the normalized points $x^k_{\mU \oplus \mW}/\|x^k_{\mU \oplus \mW}\|$ converges to a point $y^0 \in \mU \oplus \mW$.
    Without loss of generality, we can replace the sequence $(x^k)_{k \geq 1}$ by the elements corresponding to this subsequence and thus assume $\lim_{k \to \infty} x^k_{\mU \oplus \mW}/\|x^k_{\mU \oplus \mW}\| = y^0$.

    By \eqref{EQ:fullprooflowerboundonf}, we can conclude two things:
    First, we need to have $\lim_{k \to \infty} \langle w, x^k_{\mU \oplus \mW} \rangle = \infty$. Without loss of generality we can thus assume $\langle w, x^k_{\mU \oplus \mW} \rangle > 0$ for all $k$, which implies $\langle w, y^0 \rangle \geq 0$.
    Second, we need to have 
    \[
        \lim_{k \to \infty} \frac{\left\|U\frac{x^k_{\mU \oplus \mW}}{\|x^k_{\mU \oplus \mW}\|}\right\|}{\left\langle w, \frac{x^k_{\mU \oplus \mW}}{\|x^k_{\mU \oplus \mW}\|} \right\rangle} = \lim_{k \to \infty} \frac{\|Ux^k_{\mU \oplus \mW}\|}{\langle w, x^k_{\mU \oplus \mW} \rangle} = 0.
    \]
    Indeed, if $\langle w, x^k_{\mU \oplus \mW} \rangle \leq c \|Ux^k_{\mU \oplus \mW}\|$ for all $k$ and a constant $c$ independent of $k$, then, by \eqref{EQ:fullprooflowerboundonf},
    \[
        f(x^k_{\mU \oplus \mW}) \geq q(0) -  (\|\nabla q(0)\| + c) \cdot \|Ux^k_{\mU \oplus \mW}\| + \frac{\mu}{2} \|Ux^k_{\mU \oplus \mW}\|^2.
    \]
    The right hand side is globally lower bounded since it is a quadratic polynomial with positive leading coefficient, which contradicts $\lim_{k \to \infty} f(x^k_{\mU \oplus \mW}) = - \infty$.
    
    In particular, since $\left\langle w, \frac{x^k_{\mU \oplus \mW}}{\|x^k_{\mU \oplus \mW}\|} \right\rangle \leq \|w\|$ this furthermore implies 
    \[
        \|Uy^0\| = \lim_{k \to \infty} \left\|U \frac{x^k_{\mU \oplus \mW}}{\|x^k_{\mU \oplus \mW}\|}\right\| = 0 
    \]
    and thus $Uy^0 = 0$ or in other word $y^0 \in \mU^\perp$.
    Hence, $y^0 \in \mW$ and since
    \[
        \|y^0\| = \lim_{k \to \infty} \left\|\frac{x^k_{\mU \oplus \mW}}{\|x^k_{\mU \oplus \mW}\|}\right\| = 1,
    \]
    we also need to have $\langle w, y^0\rangle > 0$ (as opposed to just $\langle w, y^0 \rangle \geq 0$), which completes the proof.
\end{proof}

\subsection{\texorpdfstring{Bounding the norm of minimizers in a subspace: Proof of \Cref{LEM:fullproofboundinUandWdirection}}{Bounding the norm of minimizers in a subspace}}\label{SEC:fullproofboundinUandWdirection}
In order to prove \Cref{LEM:fullproofboundinUandWdirection}, we want to use \eqref{EQ:fullprooflowerboundonf} and show that if $\|Ux\|$ or $|\langle w, x \rangle|$ is large, then $x$ cannot be a minimizer.
We need Farkas' lemma to get a certificate for the fact that there is no $x^0$ as in \Cref{LEM:fullproofconditionforunbounded}.
\begin{proposition}[Farkas' lemma, see e.g. {\cite[Corollary 7.1e]{Schrijver1994}}]\label{PROP:Farkaslemma}
    Let $C \in \R^{M \times N}$ be a matrix and $d \in \R^M$ be a vector.
    Exactly one of the following two statements hold:
    \begin{itemize}
        \item The system $Cx \leq d$ has a solution.
        \item There is a vector $y \in \R_{\geq 0}^M$ with $C^\top y = 0$ and $d^\top y < 0$.  
    \end{itemize}
\end{proposition}
Applying this to the system from \Cref{LEM:fullproofconditionforunbounded}, we get the following.
\begin{lemma}\label{LEM:Farkaslemmaforboundedpolynomial}
    Let $f \in \Q[x]$ be a convex polynomial and let $P = \{Ax \leq b\}$ be a polyhedron.
    Let $U$ and $w$ as in \Cref{THM:quantstructure}.
    If $f$ is bounded from below on $P$, then there exist vectors $\lambda \in \R_{\geq 0}^n$ and $z \in \R^k$ such that
    \begin{equation}\label{EQ:Farkaslemmaforboundedpolynomial}
        w = A^\top \lambda + U^\top z.
    \end{equation}
    Furthermore, there exist such $\lambda$ and $z$ with $\bc(\lambda), \bc(z) \leq \poly(\enc{f},\, \bc(A))$.
\end{lemma}

\begin{proof}
    Since $f$ is bounded from below on $P$, there does not exist a vector as in \Cref{LEM:fullproofconditionforunbounded}, i.e. there is no $x^0$ with $Ax^0 \leq 0$, $Ux^0 = 0$ and $w^\top x^0 = 1$.
    By applying \Cref{PROP:Farkaslemma} to the system 
    \[
        C = \begin{bmatrix} A\\U\\-U\\w^\top\\-w^\top \end{bmatrix}, \: d = \begin{bmatrix}0\\0\\0\\1\\-1\end{bmatrix}
    \]
    we get that there is a vector $y \geq 0$ and $C^\top y = 0$ and $d^\top y < 0$. Decomposing
    \[
        y = \begin{bmatrix}\lambda'\\y^1\\y^2\\\alpha_1\\\alpha_2\end{bmatrix}
    \]
    for $\lambda' \in \R_{\geq 0}^n$, $y^1,y^2 \in \R_{\geq 0}^k$, $\alpha_1,\alpha_2 \in \R_{\geq 0}$, we get
    \[
        A^\top \lambda' + U^\top (y^1-y^2) + (\alpha_1-\alpha_2)w = 0 \text{ and } \alpha_1-\alpha_2 < 0.
    \]
    Rescaling this by $\frac{1}{\alpha_2-\alpha_1} > 0$ and defining $\lambda = \frac{1}{\alpha_2-\alpha_1} \lambda'$ and $z = \frac{1}{\alpha_2-\alpha_1} (y^1-y^2)$, we get 
    \[
        A^\top \lambda + U^\top z = w.
    \]
    By \Cref{PROP:solvelinearprogram}, we can even get $\lambda$ and $z$ with
    \[
        \bc(\lambda), \bc(z) \leq \poly(\bc(C),\, \bc(d)) = \poly(\bc(A),\, \bc(U),\, \bc(w)) = \poly(\enc{f},\, \bc(A)),
    \]
    where the last step used that $\bc(U), \bc(w) \leq \poly(\enc{f})$ by \Cref{THM:quantstructure}.
    Note that, since rescaling~$y$ by positive scalars does not change feasibility, we can replace $d^\top y < 0$ by $d^\top y = -1$  (thus making it a linear program as in \Cref{PROP:solvelinearprogram}).
\end{proof}

We can now prove \Cref{LEM:fullproofboundinUandWdirection}.

\begin{proof}[Proof of \Cref{LEM:fullproofboundinUandWdirection}]
    Let $\lambda \in \R^n_{\geq 0}$ and $z \in \R^k$ as in \Cref{LEM:Farkaslemmaforboundedpolynomial}.
    For any $x \in \R^n$ we get
    \[
        \langle w, x \rangle = \langle A^\top \lambda, x\rangle + \langle U^\top z, x \rangle = \langle \lambda, Ax\rangle + \langle z, Ux \rangle.
    \]
    For $x \in P$ we have $Ax \leq b$, which together with $\lambda \geq 0$ implies
    \[
        \langle \lambda, Ax\rangle \leq \lambda^\top b.
    \]
    Thus, we get for any $x \in P$    
    \begin{equation}\label{EQ:fullproofboundwcompintermsofUcomp}
        |\langle w, x \rangle| \leq \|\lambda \| \cdot \|b\| + \|z\| \cdot \|Ux\|.
    \end{equation}
    Using \eqref{EQ:fullproofboundwcompintermsofUcomp} in \eqref{EQ:fullprooflowerboundonf}, we get for all $x \in P$
    \begin{align*}
         f(x) &\geq q(0) - \|\nabla q(0)\| \cdot \|Ux\| - \|\lambda \| \cdot \|b\| - \|z\| \cdot \|Ux\| + \frac{\mu}{2} \|Ux\|^2\\
         &= (q(0) - \|\lambda \| \cdot \|b\|) + (- \|\nabla q(0)\| - \|z\|) \cdot \|Ux\| +  \frac{\mu}{2} \|Ux\|^2
    \end{align*}
    Fix $a \in P$ with $\|a\| \leq 2^{\poly(\enc{P})}$ (such a point exists by \Cref{PROP:solvelinearprogram}).
    We are interested when the lower bound is at least $f(a)$, i.e. when
    \[
    (q(0) - \|\lambda \| \cdot \|b\| - f(a)) + (- \|\nabla q(0)\| - \|z\|) \cdot \|Ux\| +  \frac{\mu}{2} \|Ux\|^2 \geq 0
    \]
    Since this is a quadratic polynomial in $\|Ux\|$ with positive leading coefficient, as $\|Ux\| \to \infty$ this is positive.
    Note that for $x=a$ the lower bound needs to be non-positive, i.e. this polynomial in $\|Ux\|$ has at least one root.
    Thus, if $\|Ux\|$ is larger than the largest root, it needs to be positive.
    That is, whenever
    \[
        \|Ux\| > \frac{ \|\nabla q(0)\| + \|z\| + \sqrt{(\|\nabla q(0)\| + \|z\|)^2 + 2\mu (f(a) + \|\lambda \| \cdot \|b\| - q(0))}}{\mu}
    \]
    for some $x \in P$, then we have $f(x) > f(a)$.
    Furthermore, combining this with \eqref{EQ:fullproofboundwcompintermsofUcomp}, whenever
    \[
        |\langle w, x \rangle| > \|\lambda \| \cdot \|b\| + \|z\| \cdot \frac{ \|\nabla q(0)\| + \|z\| + \sqrt{(\|\nabla q(0)\| + \|z\|)^2 + 2\mu (f(a) + \|\lambda \| \cdot \|b\| - q(0))}}{\mu}
    \]
    for some $x \in P$, then we also have $f(x) > f(a)$.
    
    Thus, if $\|Ux\|$ or $|\langle w, x \rangle|$ are large, then $f(x) > f(a)$. We now want to quantify this in terms of the bit length of the input.
    We do this term by term:
    \begin{itemize}
        \item By \Cref{THM:quantstructure}, we have $\bc(q) \leq \poly(\enc{f})$ and thus $|q(0)|, \|\nabla q(0)\| \leq 2^{\poly(\enc{f})}$.
        \item By \Cref{THM:quantstructure}, we have $\mu \geq 2^{-\poly(\enc{f})}$.
        \item By \Cref{LEM:Farkaslemmaforboundedpolynomial}, we have $\bc(z) \leq \poly(\bc(A),\, \bc(U),\, \bc(w)) \leq \poly(\enc{f},\, \bc(A))$ and thus $\|z\| \leq 2^{\poly(\enc{f}, \,\enc{P})}$.
        \item By \Cref{LEM:Farkaslemmaforboundedpolynomial}, we have $\bc(\lambda) \leq \poly(\bc(A),\, \bc(U),\, \bc(w)) \leq \poly(\enc{f},\, \bc(A))$ and thus $\|\lambda\| \leq 2^{\poly(\enc{f}, \,\enc{P})}$.
        \item Since we have $\|a\| \leq 2^{\poly(\enc{P})}$, we can bound $|f(a)| \leq 2^{\poly(\enc{f})} \cdot \|a\|^d$ and thus we get $|f(a)| \leq 2^{\poly(\enc{f}, \,\enc{P})}$.
        \item We have $\|b\| \leq 2^{\poly(\enc{P})}$.
    \end{itemize}
    Putting this all together, we get that for all $x \in P$ we have
    \begin{equation}\label{EQ:fullproofUxorAtoplambdaxlarge}
        \|Ux\| > 2^{\poly(\enc{f}, \,\enc{P})} \quad \text{or} \quad |\langle w, x \rangle| > 2^{\poly(\enc{f}, \,\enc{P})} \quad \Longrightarrow \quad f(x) > f(a).
    \end{equation}
    To complete the proof, it remains to connect $\|Ux\|$ and $|\langle w, x\rangle|$ to $x_{\mU \oplus \mW}$.
    Note that we have
    \[
        x_{\mU \oplus \mW} = \sum_{i=1}^k \frac{(Ux)_i}{\|U_i\|^2}U_i + \frac{\langle w, x\rangle}{\|w\|^2}w,
    \]
    where $U_i$ are the rows of $U$ (that are orthogonal and span $\mU$).
    Thus, we have 
    \[
        \|x_{\mU \oplus \mW}\| \leq \sum_{i=1}^k \frac{|(Ux)_i|}{\|U_i\|} + \frac{|\langle w, x\rangle|}{\|w\|}.
    \]
    First, since $\bc(U) \leq \poly(\enc{f})$ by \Cref{THM:quantstructure} we have $\|U_i\| \geq 2^{-\poly(\enc{f})}$.
    Furthermore, we have that $\sum_{i=1}^k |(Ux)_i| = \|Ux\|_1 \leq \sqrt{k} \|Ux\|$.
    Thus, we get that
    \[
        \sum_{i=1}^k \frac{|(Ux)_i|}{\|U_i\|} \leq 2^{\poly(\enc{f})} \|Ux\|.
    \]
    Second, by \Cref{THM:quantstructure}, we have $\bc(w) \leq \poly(\enc{f})$ and thus $\|w\| \geq 2^{-\poly(\enc{f})}$.
    Together, this gives
    \[
        \|x_{\mU \oplus \mW}\| \leq 2^{\poly(\enc{f})} (\|Ux\|_2 + |\langle w, x\rangle|)
    \]
    and hence, by \eqref{EQ:fullproofUxorAtoplambdaxlarge}, we can conclude that, for all $x \in P$,
    \[
        \|x_{\mU \oplus \mW}\| > 2^{\poly(\enc{f}, \,\enc{P})} \Longrightarrow f(x) > f(a).\qedhere
    \]
\end{proof}

\subsection{\texorpdfstring{Finding a minimizer of small norm: Proof of \Cref{LEM:fullproofboundinVdirection}}{Finding a minimizer of small norm}}\label{SEC:fullproofboundinVdirection}
It remains to prove \Cref{LEM:fullproofboundinVdirection}, which shows that we can lift a point $x^* \in P$ with small norm in some subspace $\mZ$ to a point $x' \in P$, whose norm in the entire space is small.

\begin{proof}[Proof of \Cref{LEM:fullproofboundinVdirection}]
    Consider the polyhedron
    \begin{align*}
        P' &= \{x \in \R^n : Ax \leq b, \: \langle x, x^i \rangle = \langle x^*, x^i \rangle \: \forall 1 \leq i \leq k\}\\
        &= \left\{ x \in \R^n : A' x \leq b'\right\}
    \end{align*}
    for
    \[
        A' = \begin{bmatrix} A \\ \phantom{-}{x^1}^\top \\ - {x^1}^\top \\ \vdots \\ \phantom{-}{x^k}^\top \\ -{x^k}^\top \end{bmatrix} \quad \text{and} \quad b' = \begin{bmatrix} b\\\phantom{-}\langle x^*, x^1 \rangle\\-\langle x^*, x^1 \rangle\\ \vdots \\ \phantom{-}\langle x^*, x^k \rangle\\-\langle x^*, x^k \rangle \end{bmatrix}.
    \]
    Note that $x' \in P'$ if and only if $x' \in P$ and $x^*_\mZ = x'_\mZ$.
    Thus, it remains to show that there is a $x' \in P'$ with $\|x'\| \leq \poly(\|x^*_\mZ\|,\, 2^{\poly(\enc{P}, \,\bc(\mZ))})$.
    
    Notice that we cannot apply \Cref{PROP:solvelinearprogram} immediately since the vector $b'$ might not be rational.
    Instead, we want to pick a vertex of $P'$ and argue that it has small norm (even though it might also not be rational).
    However, it could be that $P'$ has no vertices.
    Namely, $P'$ has no vertices if and only if $\rank(A') < n$ or in other words $\ker(A') \neq \{0\}$ \cite[section 8.5]{Schrijver1994}.
    Let $y^1, \ldots, y^\ell$ be a basis for $\ker(A')$.
    We have $\ell \leq n$ and since the $y^j$ are solutions to the linear system $A' y = 0$, they satisfy $\bc(y^j) \leq \poly(\bc(A),\, \bc(x^1),\, \ldots,\, \bc(x^k)) \leq \poly(\bc(A),\, \bc(\mZ))$ \cite[Corollary 3.2d]{Schrijver1994}.
    Consider the polyhedron
    \begin{align*}
        P'' &= \{x \in \R^n : Ax \leq b, \: \langle x, x^i \rangle = \langle x^*, x^i \rangle \: \forall 1 \leq i \leq k, \: \langle x, y^j\rangle = 0 \: \forall 1 \leq j \leq \ell\}\\
        &= \left\{ x \in \R^n : A'' x \leq b''\right\},
    \end{align*}
    where
    \[
        A'' = \begin{bmatrix} A \\ \phantom{-}{x^1}^\top \\ - {x^1}^\top \\ \vdots \\ \phantom{-}{x^k}^\top \\ - {x^k}^\top \\ \phantom{-}y^1 \\ -y^1 \\ \vdots \\ \phantom{-}y^\ell \\ -y^\ell \end{bmatrix} \quad \text{and} \quad b'' = \begin{bmatrix} b\\\phantom{-}\langle x^*, x^1 \rangle\\-\langle x^*, x^1 \rangle \\ \vdots \\ \phantom{-}\langle x^*, x^k \rangle\\-\langle x^*, x^k \rangle \\ 0 \\ 0 \\ \vdots \\ 0 \\ 0\end{bmatrix}.
    \]
    We claim that the projection ${x^*}' \coloneqq x^*_{\spann{y^1, \ldots, y^\ell}^\perp}$ of $x^*$ to $\spann{y^1, \ldots, y^\ell}^\perp$ is in $P''$, i.e. that $P''$ is feasible.
    Since $x^* \in P'$ and $y^1, \ldots, y^\ell \in \ker(A')$, we have $A'{x^*}' \leq b'$.
    Clearly, also $\langle {x^*}', y^j \rangle = 0$, so we indeed have ${x^*}' \in P''$.
    Furthermore, since $P'' \subseteq P'$, it is sufficient to find a point $x' \in P''$ that satisfies $\|x'\| \leq \poly(\|x^*_\mZ\|,\, 2^{\poly(\enc{P}, \,\bc(\mZ))})$.
    
    Note that since the $y^j$ generate $\ker(A')$, we have $\ker(A'') = \{0\}$.
    Thus, $\rank(A'') = n$ and $P''$ has a vertex $x'$.
    This vertex is the solution to a subsystem
    \[
        \hat{A}x' = \hat{b},
    \]
    where $\hat{A}$ contains $n$ linearly independent rows of $A''$ and $\hat{b}$ contains the corresponding elements of $b''$ (see e.g. \cite[equation (23)]{Schrijver1994}).
    This means that we have
    \[
        x' = \hat{A}^{-1} \hat{b}
    \]
    and thus also
    \[
        \|x'\| \leq \|\hat{A}^{-1}\| \|\hat{b}\|.
    \]
    Note that $\bc(\hat{A}) \leq \bc(A'') \leq \poly(\bc(A),\, \bc(\mZ))$.
    We also have $\bc(\hat{A}^{-1}) \leq \poly(\bc(\hat{A}))$ \cite[Corollary 3.2a]{Schrijver1994}.
    Thus, we get 
    \[
        \|\hat{A}^{-1}\| \leq \sqrt{n} \|\hat{A}^{-1}\|_F \leq 2^{\poly(\bc(\hat{A}))} \leq 2^{\poly(\bc(A),\, \bc(\mZ))}.
    \]
    For this, note that $n \leq \bc(\hat{A})$.
    Furthermore, we have
    \[
        \|\hat{b}\|^2 \leq \|b\|^2 + \sum_{i = 1}^k \langle x^*, x^i\rangle^2 \leq \|b\|^2 + \sum_{i=1}^k \|x^*_\mZ\| \|x^i\| \leq \poly(\|x^*_\mZ\|,\, 2^{\poly(\bc(b), \bc(\mZ))}).
    \]
    Combining these, we get
    \[
        \|x'\| \leq \poly(\|x^*_\mZ\|,\, 2^{\poly(\enc{P}, \,\bc(\mZ))}),
    \]
    which completes the proof.
\end{proof}

\section*{Acknowledgments}
\addcontentsline{toc}{section}{Acknowledgments}
We thank Amir Ali Ahmadi for bringing the problem of efficient minimization of convex polynomials to our attention. We thank Rares-Darius Buhai, Jarosław Błasiok,  Hongjie Chen, Jingqiu~Ding and Yiding Hua for valuable discussions at an early stage of the project.
We thank Abraar~Chaudhry for a discussion on quadratic lower bounds for convex polynomials.
We thank Kevin~Shu for a discussion on polynomials with vanishing Hessian determinant, and in particular for pointing us to the Schwartz-Zippel lemma. 

\smallskip\noindent
This research was supported by the Swiss National Science Foundation (SNSF), grant no. 10004947.

\appendix
\section{Linear algebra in the bit model}
\label{APP:LA}
The following proposition shows that one can compute an orthogonal basis of small bit size for the inverse image of a (one-dimensional) subspace under a rational matrix, and for its complement. This captures the linear algebra required in our proof of~\Cref{THM:quantstructure}. 
\begin{proposition} \label{PROP:polytimeLA}
    Let $A \in \Q^{M \times N}$, $M \geq N$, and $b \in \Q^M$. Define the subspaces
    \[
        \mL = \{v \in \R^N : A \cdot v \in \spann{b} \}, \quad \mU = \mL^\perp.
    \]
    We then have $\mL = \ker(A) \oplus \spann{w}$, where $w \in \ker(A)^\perp$ is either zero or $A \cdot w = b$.    
    We can compute the vector $w$; an orthogonal basis $v_1, \ldots, v_\ell$ for $\ker (A)$; an orthogonal basis $u_1, \ldots, u_k$ for~$\mU$ in polynomial time in $\bc(A) + \bc(b)$. In particular, these have polynomial bit size in $\bc(A) + \bc(b)$.
\end{proposition}
\Cref{PROP:polytimeLA} follows from standard results on linear system solving and Gram-Schmidt orthogonalization in the bit model. We include a brief proof for completeness. We rely on the \emph{Hermite normal form} of an integer matrix, which is an integer analog of the echelon form. The same conclusion could be reached using a polynomial-time algorithm to compute the echelon form (over~$\Q$), as described, e.g., in~\cite{Schrijver1994}.

\begin{lemma}[Hermite normal form] \label{LEM:HNF}
Let $A \in \Z^{M \times N}$, $M \geq N$ be of column rank $R$. We can compute in time polynomial in $\bc(A)$ a decomposition
\[
    A \cdot 
    \begin{bmatrix}
    \begin{array}{c|c}
    U & K
    \end{array}
    \end{bmatrix} = 
    \begin{bmatrix}
    \begin{array}{c|c}
    H & 0
    \end{array}
    \end{bmatrix},
\]
where $U \in \Z^{N \times R}$, $K \in \Z^{N \times (N - R)}$ satisfy $[U \mid K] \in \mathrm{GL}_N(\Z)$, and $H \in \Z^{M \times R}$ is lower triangular. In particular these matrices have bit length polynomial in $\bc(A)$.    
\end{lemma}
\begin{proof}
    Apply \cite[Proposition 6.3]{Storjohann2000} and \cite[Proposition 6.6]{Storjohann2000} to $A^\top$.
\end{proof}

\begin{lemma}[Gram-Schmidt orthogonalization] \label{LEM:LAGS}
    Let $v_1, \ldots, v_k \in \Q^{N}$. Then, there exist pairwise orthogonal vectors $u_1, \ldots, u_k \in \Q^N$, such that
    \[
        \spann{u_1, \ldots, u_i} = \spann{v_1, \ldots, v_i} \quad (\forall i \leq k).
    \]
    Moreover, these vectors can be computed in polynomial time in $N$, $k$ and $B \coloneqq \max_{i}\bc(v_i)$. 
\end{lemma}
\begin{proof}
    The only thing to check is that the recursively defined coefficients that appear in the Gram-Schmidt procedure do not grow too large. Such an analysis is carried out, e.g., in~\cite{Lenstra1982} or~\cite{ErlingssonKaltofenMusser:GramSchmidt}.
\end{proof}

\begin{proof}[Proof of~\Cref{PROP:polytimeLA}]
    After multiplication by an integer of bit length $\poly(\bc(A))$, we may assume without loss of generality that $A$ is an integer matrix. Let $R$ be the column rank of $A$.
    Compute the Hermite normal form $A \cdot [U \mid  K] = [H \mid 0]$ as in~\Cref{LEM:HNF}. Note that the columns of $[U \mid K]$ form a basis for $\R^N$, and that the columns of $K$ form a basis for~$\ker(A)$. The lower triangular system $H \cdot y = b$ can be solved via forward substitution; this either yields a solution $y \in \Q^{R}$ (of bit length at most $\poly(\bc(H))$), or shows that no solution exists at all. In the former case, $\hat w \coloneqq U\cdot y$ satisfies $A \cdot \hat w = b$ (and $\bc(\hat w) \leq \poly(\bc(U),\, \bc(H))$. It remains to orthogonalize: 
    Apply \Cref{LEM:LAGS} to the columns of $K$, $\hat w$, and the columns of $U$, \emph{in that order}. This yields orthogonal vectors  $v_1, \ldots, v_{N-R}, \, w, \, u_{1}, \ldots, u_{R}$ (of appropriately bounded bit length). Note that $v_1,\ldots,v_{N-R}$ span $\ker(A)$, that $v_1,\ldots,v_{N-R}, w$ span $\mathcal{L}$, and that $u_1, \ldots, u_{R}$ span $\mathcal{L}^\perp$.
    Note that either $w$, or exactly one of the $u_i$ will be zero. Discarding that vector completes the proof.
\end{proof}

\section{A polynomial with identically vanishing hessian} \label{APP:HessianCounterexample}
We provide an example adapted from~\cite{Gordan1876, Garbagnati2009} of a polynomial whose Hessian is everywhere singular, but which does not have a direction of linearity. 
Let $p(x) = x_1x_4^2 + 2x_2x_4x_5 + x_3x_5^2$.
Then,
\[
\nabla^2 p (x) = \begin{pmatrix}
0 & 0 & 0  & 2x_{4} & 0 \\
0 & 0 & 0  & 2x_{5} & 2x_{4} \\
0 & 0 & 0  & 0 & 2x_{5} \\
2x_{4} & 2x_{5} & 0 & 2x_{1} & 2x_{2} \\
0 & 2x_{4} & 2x_{5} & 2x_{2} & 2x_{3}
\end{pmatrix}.
\]
One may verify that the Hessian of $p$ is singular for every $x \in \R^n$. On the other hand, we have
\[
    \nabla p(x) = (x_4^2,~2x_4x_5,~ x_5^2,~ 2x_1x_4 + 2x_2x_5,~ 2x_2x_4 + 2x_3x_5),
\]
and so $p$ has no direction of linearity (its directional derivatives are not affinely dependent). But, we do have an algebraic dependency of degree $2$, namely
\[
    \frac{\partial p}{\partial x_1} \cdot \frac{\partial p}{\partial x_3} = x_4^2 \cdot x_5^2 = \frac{1}{4}(2x_4x_5)^2 = \frac{1}{4} \left(\frac{\partial p}{\partial x_2} \right)^2.
\]

\section{Witnesses for polynomial programs}
\label{APP:Complexity}
In this appendix, we give a rational univariate convex polynomial of degree 4 that has an irrational minimizer (see \Cref{EX:convexdegree4irrationalminimizer}). 
This shows that for the decision problem~\eqref{EQ:decision} for convex polynomials of degree at least 4, $x^*$ is not always a compact witness.
We also give a rational univariate convex polynomial of degree 6 that has minimum value 0, but attains this at exactly one \emph{irrational} point (see \Cref{EX:convexdegree6irrationalminimizerrationalminimum}).
This shows that for convex polynomials of degree at least 6, for some polynomials there is no compact witness at all.
Finally, we give a proof that such a (univariate) polynomial cannot exist for degree 4.
We show that for a rational univariate convex polynomial of degree 4, if the minimum value is rational, then also the minimizer is rational (see \Cref{LEM:convexdegree4ifminimumrationalthenminimizerrational}).
Thus, for convex polynomials of degree 4, it could be that there is always a compact witness even though the minimizer might not be:
If the minimizer is irrational, then the set $\{x: f(x) \leq 0\}$ does not contain just the minimizer since the minimum value is not $0$ (it is irrational).
Thus, there will at least be a rational point in this set and it is unclear whether there might always be a compact witness.

\begin{example}\label{EX:convexdegree4irrationalminimizer}
    Consider the polynomial
    \[
        f(x) = x^4 + x.
    \]
    The second derivative is positive for all $x$, thus $f$ is convex.
    The (unique) minimizer of $f$ is the point $x^*$ that satisfies $4{x^*}^3 + 1 = 0$, i.e., $x^* = -1/\sqrt[3]{4}$, which is irrational.
\end{example}

\begin{example}\label{EX:convexdegree6irrationalminimizerrationalminimum}
    Consider the polynomial
    \[
        f(x) = (x^3 + x + 1)^2 = x^6 + 2x^4 + 2x^3 + x^2 + 2x + 1.
    \]
    The second derivative of $f$ is given by
    \[
        f''(x) = 30 x^4 + 24 x^2 + 12 x + 2 = 30 x^4 + 6x^2 + 2 \cdot (3x + 1)^2 \geq 0 \quad (\forall x \in \R)
    \]
    and hence $f$ is a convex polynomial.

    We now want to determine the minimal value and the minimizer of $f$.
    Note that $f$ is non-negative.
    To determine the minimum and the minimizer, we want to examine the term ${g(x) = x^3 + x + 1}$.
    As a cubic polynomial it has at least one real root.
    Thus, the minimum value of $f$ is $0$ and the minimizer(s) of $f$ are exactly the root(s) of $g$.
    
    Cardano's formula states that a cubic polynomial $x^3 + px + q$ with $\frac{q^2}{4} + \frac{p^3}{27} > 0$ has in fact exactly one real root.
    Furthermore, by the rational root theorem, if there is a rational root it is an integer divisor of $q$.
    This implies that $g$ (for which $p=q=1$) has exactly one real root $x^*$.
    Since $\pm 1$ (the only divisors of $q=1$) are not roots, this root is not rational, i.e. $x^* \in \R \backslash \Q$.
    This $x^*$ is the (unique) minimizer of $f$.
    In fact, we can even compute $x^*$ explicitly:
    \[
        x^* = \sqrt[3]{-\frac{1}{2}+\sqrt{\frac{1}{4}+\frac{1}{27}}} + \sqrt[3]{-\frac{1}{2}-\sqrt{\frac{1}{4}+\frac{1}{27}}}.
    \]
    
    Hence, $f$ is a convex polynomial with minimum value $0$, but this is attained at exactly one (irrational) point $x^*$.
\end{example}

\begin{lemma}\label{LEM:convexdegree4ifminimumrationalthenminimizerrational}
    Let $f \in \Q[x]$ be a convex (univariate) polynomial of degree $4$.
    Assume the minimum value of $f$ is rational.
    Then also the minimizer of $f$ is rational.
\end{lemma}
\begin{proof}
    Let $x^* \in \R$ be the minimizer of $f$.
    By convexity, $x^*$ is the unique minimizer and the unique critical point of $f$.
    Consider the two polynomials (recall that by assumption $f(x^*) \in \Q$)
    \[
        g_1(x) = f(x) - f(x^*) \in \Q[x] \quad \text{and} \quad g_2(x) = f'(x) \in \Q[x].
    \]
    We have that $x^*$ is the unique (real) root of these two polynomials.
    Consider $m \in \Q[x]$ the monic polynomial of minimal degree that has $x^*$ as a root.
    Since $x^*$ is a root of $g_1$ and $g_2$, we have $m \mid g_1$ and $m \mid g_2$.\footnote{Formally, consider the field extension $\Q(x^*)/\Q$. Since $g_1(x^*) = 0$ and $g_1 \in \Q[x]$, $x^*$ is algebraic over $\Q$. Then, $m$ is the minimal polynomial of $x^*$ over $\Q$, i.e. the unique monic irreducible polynomial $m \in \Q[x]$ that has $x^*$ as a root \cite[Theorem A-3.87]{RotmanAlgebra}. Furthermore, $m$ generates the ideal of all polynomials in $\Q[x]$ that have $x^*$ as a root. Thus, $m$ divides any polynomial $g$ that has $x^*$ as a root \cite[Proof of Theorem A-3.87]{RotmanAlgebra}.}
    Thus, there are polynomials $h_1, h_2 \in \Q[x]$ such that $g_1 = m \cdot h_1$ and $g_2 = m \cdot h_2$.
    
    Note that $\deg(h_1) = \deg(h_2) + 1$.
    Thus one of $\deg(h_1)$ and $\deg(h_2)$ is odd.
    Let $i \in \{1,2\}$ be the index such that $\deg(h_i)$ is odd.
    Then $h_i$ has a real root.
    Since $x^*$ is the unique real root of $g_i$, this root needs to be $x^*$.
    But then we need to have $m \mid h_i$ and there is a polynomial $\Tilde{h}_i$ such that $g_i = m^2 \cdot \Tilde{h}_i$.
    This implies $4 \geq \deg(g_i) \geq 2 \deg(m)$ and thus $\deg(m) \leq 2$.
    
    If $\deg(m) = 1$, then $m(x) = x + a_0$ and thus $x^* = -a_0 \in \Q$.
    If $\deg(m) = 2$, then $m(x) = x^2 + a_1 x + a_0$.
    The roots of this polynomial are $\frac{-a_1 \pm \sqrt{a_1^2-4a_0}}{2}$.
    Since $x^*$ is the unique (real) root of $g_1$ (and thus of~$m$), we need to have that $a_1^2-4a_0 = 0$ and thus $x^* = \frac{-a_1}{2} \in \Q$, which completes the proof.
\end{proof}

\section{Ellipsoid method}
\label{APP:Ellipsoid}
In this appendix, we give a proof of \Cref{PROP:applicationellipsoid} about minimizing a convex polynomial over a polyhedron.
\PROPapplicationellipsoid*

All the ideas presented in this appendix are standard results on the ellipsoid method and can for example be found in \cite{GrötschelLovaszSchrijver:ellipsoid, Vishnoi:convexoptimization}.
We combine these results to get the exact statement of \Cref{PROP:applicationellipsoid}.
We include a full proof for completeness and in order to carry out the bit complexity analysis.

\subsection{Additional preliminaries on the ellipsoid method}
In this section, we discuss the necessary preliminaries to apply the ellipsoid method to our case.
The section is adapted from different definitions and theorems from \cite{GrötschelLovaszSchrijver:ellipsoid}.
We first define what a strong separation oracle is.
\begin{definition}[Strong separation oracle]
    A strong separation oracle for a convex set $K \subseteq \R^n$ is an oracle that, given as input $y \in \Q^n$ either asserts $y \in K$ or finds a separating hyperplane $c \in \Q^n$ such that $\max_{x \in K} c^\top x < c^\top y$.
\end{definition}

\begin{remark}
    We always assume that given an input $y \in \Q^n$, the bit length of the output of the oracle is at most polynomial in $\bc(y)$ for a fixed polynomial. See also \cite[Assumption 1.2.1]{GrötschelLovaszSchrijver:ellipsoid}.
\end{remark}

For a polyhedron, there is always a strong separation oracle that runs in polynomial time in the bit length of the input and in $\enc{P}$.

\begin{lemma}[{\cite[Example 2.16]{GrötschelLovaszSchrijver:ellipsoid}}]\label{LEM:seporacleforP}
    For a polyhedron $P = \{Ax \leq b\}$ there is a strong separation oracle that runs on an input of bit length $k$ in time polynomial in $k$ and $\enc{P}$.
\end{lemma}

The idea for this separation oracle is to check for all constraints of $P$ whether $y$ satisfies them.
If all constraints are satisfied, then $y \in P$.
Otherwise, any violated constraint is a separating hyperplane.
We now state a result on how we can apply the ellipsoid method to solve feasibility problems given strong separation oracles.

\begin{proposition}[{\cite[Theorem 3.2.1]{GrötschelLovaszSchrijver:ellipsoid}}]\label{PROP:ellipsoidforconvexset}
    There is an algorithm with the following guarantees:
    Given as input $r > 0$, $R > 0$ and a strong separation oracle to a closed convex set $K \subseteq B_R(0)$, the algorithm outputs outputs one of the following:
    \begin{enumerate}
        \item a vector $a \in K$;\footnote{We note that the algorithm in \cite[Theorem 3.2.1]{GrötschelLovaszSchrijver:ellipsoid}, if given a strong separation oracle (instead of a weak one), in fact outputs $y \in K$ (as opposed to just $d(y,K) \leq r$).}
        \item an ellipsoid $E$ such that $K \subseteq E$ and $\vol(E) \leq r$.
    \end{enumerate}
    In particular, if $\vol(K) > r$, then the algorithm outputs a vector in $K$.
    The number of oracle calls and the runtime\footnote{Here, the runtime does not include the oracle call, but it does include the time needed to write down the input for the oracle and read the output of the oracle.} of the algorithm are polynomial in $n$, $\log(1/r)$ and $\log(R)$.
\end{proposition}
\begin{remark}
    When we apply \Cref{PROP:ellipsoidforconvexset}, the oracle runs in polynomial time in $\enc{f}$, $\enc{P}$, $k$ on an input of bit length $k$.
    Therefore, the total runtime of the algorithm will be polynomial in $\enc{f}$, $\enc{P}$, $\log(1/r)$ and $\log(R)$.
\end{remark}

To apply the statement above, we also need the following two statements that will allow us to reduce to the case where the polynomial has volume (i.e. the full-dimensional case).

\begin{proposition}[{\cite[Theorems 6.4.9 and 6.5.5]{GrötschelLovaszSchrijver:ellipsoid}}]\label{PROP:getaffinehullofpolyhedron}
    There is an algorithm with the following guarantees:
    Let $P = \{Ax \leq b\} \subseteq \R^n$ be a polyhedron and let $\varphi$ be a bound on the maximum bit length of any constraint of $P$.
    Given as input $n$, $\varphi$ and a strong separation oracle to $P$, the algorithms output affinely independent vectors $v_1, \ldots, v_k$ such that $\aff(P) = \aff(\{v_1, \ldots, v_k\})$.
    The number of oracle calls and the runtime of the algorithm is polynomial in $n$ and $\varphi$.    
\end{proposition}

\begin{proposition}[{\cite[Lemmas 3.1.33 and 3.1.35]{GrötschelLovaszSchrijver:ellipsoid}}]\label{PROP:volumefulldimensionalcase}
    Let $P=\{Ax \leq b\}$ be a polyhedron. If $P$ is full-dimensional, then
    \[
        \vol(P) \geq 2^{-\poly(\enc{P})}.
    \]
    In fact, there exist $v_1, \ldots, v_{n+1} \in P$ with $\|v_i\| \leq 2^{\poly(\enc{P})}$ such that
    \[
        \vol(\conv(\{v_1, \ldots, v_{n+1}\}) \geq 2^{-\poly(\enc{P})}.
    \]
\end{proposition}

\subsection{\texorpdfstring{Proof of \Cref{PROP:applicationellipsoid}}{Proof}}
In order to prove \Cref{PROP:applicationellipsoid} we want to reduce to the full-dimensional case, for which we then can apply \Cref{PROP:ellipsoidforconvexset}.
We first prove the full-dimensional case of \Cref{PROP:applicationellipsoid}.
\begin{proof}[{Proof of \Cref{PROP:applicationellipsoid} (full-dimensional case)}]
    Consider the set $F_\tau \coloneqq \{x \in \R^n: f(x) \leq \tau\}$.
    Our goal is to apply \Cref{PROP:ellipsoidforconvexset} to the sets $P \cap B_R(0) \cap F_\tau$ for different values of $\tau$.
    Thus, we need to give a strong separation oracle for this set and find a value for $r$ such that $\vol(P \cap B_R(0) \cap F_\tau) \geq r$.

    \paragraph{Strong separation oracle.}
    By \Cref{LEM:seporacleforP} we have a strong separation oracle for $P$.
    Clearly, there is also a strong separation oracle for $B_R(0)$ (check if $\|y\| \leq R$; if not, $y$ is a separating hyperplane) and thus it remains to argue that we can get a strong separation oracle for $F_\tau$.
    Given $y \in \Q^n$, we can determine whether $f(y) \leq \tau$ and thus whether $y \in F_\tau$. If $y \not\in F_\tau$, then the gradient $\nabla f(y)$ is a separating hyperplane.
    This is true since $f(x) - f(y) \geq \nabla f(y)^\top (x-y)$ for all $x \in \R^n$ and thus for $x \in F_\tau$ we have $f(x) \leq \tau < f(y)$ and thus $\max_{x \in F_\tau}\nabla f(y)^\top x < \nabla f(y)^\top y$.
    Thus, we have strong separation oracles for $P$, $B_R(0)$ and $F_\tau$ and thus also a strong separation oracle for $P \cap B_R(0) \cap F_\tau$.
    This separation oracle runs in time $\poly(\enc{f},\, \enc{P},\, \bc(\tau),\, \log(R),\, \bc(y))$ if given as input a point $y \in \Q^n$.

    \paragraph{Bound on volume.}
    Next, we want to bound $\vol(P \cap B_R(0) \cap F_\tau)$.
    Let $x^*$ be a minimizer of~$f$ on $P \cap B_R(0)$ and let $f^*=f(x^*)$.
    Note that $f$ is Lipschitz on $B_R(0)$ with Lipschitz constant ${L = 2^{\poly(\enc{f},\, \log(R))}}$.
    If $\tau = f^* + c$, then we want to argue that we can lower bound the volume in terms of $c$.
    Define $c' = c/L$.
    Since $P$ is full-dimensional, by \Cref{PROP:volumefulldimensionalcase}, there are ${v_1, \ldots, v_{n+1} \in P}$ with $\|v_i\| \leq 2^{\poly(\enc{P})}$ such that
    \[
        \vol(\conv(\{v_1, \ldots, v_{n+1}\})) \geq 2^{-\poly(\enc{P})}.
    \]
    Without loss of generality, we assume $R \geq \|v_i\|$ for all $i$ such that $v_i \in P \cap B_R(0)$.\footnote{We can do this because $\|v_i\| \leq 2^{\poly(\enc{P})}$. Thus, if $R < 2^{\poly(\enc{P})}$, we can replace $R$ by $R' = 2^{\poly(\enc{P})}$, which satisfies $\log(R') = \poly(\enc{f})$.}
    Let $d_i = \|x^* - v_i\|$ and let $d_{\max} = \max_{i \in \{1, \ldots n+1\}} d_i$.
    We have, for some $i \in \{1, \ldots, n\}$,
    \[
        d_{\max} = \|x^* - v_i\| \leq \|x^*\| + \|v_i\| \leq R + 2^{\poly(\enc{P})}.
    \]
    If $d_{\max} \leq c'$, then, by Lipschitzness of $f$, we have 
    \[
        f(v_i) \leq f(x^*) + L \|v_i - x^*\| \leq f^* + L \cdot c' = f + c = \tau.
    \]
    Hence, $v_i \in P \cap B_R(0) \cap F_\tau$ and thus
    \[
        \vol(P \cap B_R(0) \cap F_\tau) \geq \vol(\conv(\{v_1, \ldots, v_{n+1}\})) \geq 2^{-\poly(\enc{P})}.
    \]
    If $d_{\max} > c'$, consider the points $\hat{v}_i \coloneqq x^* + \frac{c'}{d_{\max}}(v_i - x^*) \in P \cap B_R(0)$ (since $x^*, v_i \in P \cap B_R(0)$).
    Then, we have, again by Lipschitzness of $f$,
    \[
        f(\hat{v}_i) \leq f(x^*) + L \left\|\frac{c'}{d_{\max}}(v_i - x^*)\right\| \leq f^* + L \cdot c' \cdot \frac{d_i}{d_{\max}} \leq f^* + L \cdot c' = \tau
    \]
    and hence $\hat{v}_i \in F_\tau$.
    Thus, $\conv(\{\hat{v}_1, \ldots, \hat{v}_{n+1}\}) \subseteq P \cap B_R(0) \cap F_\tau$.
    We have
    \begin{align*}
        \vol\left(\conv\left(\left\{\hat{v}_1, \ldots, \hat{v}_{n+1}\right\}\right)\right) &= \vol\left(\conv\left(\left\{x^* + \frac{c'}{d_{\max}}(v_i - x^*) : i \in \{1, \ldots, n+1\}\right\}\right)\right)\\
        &= \left(\frac{c'}{d_{\max}}\right)^n \vol\left(\conv\left(\left\{v_i : i \in \{1, \ldots, n+1\}\right\}\right)\right)\\
        &\geq \left(\frac{c'}{d_{\max}}\right)^n \cdot 2^{-\poly(\enc{P})}.
    \end{align*}
    Thus, we get that 
    \[
        \vol(P \cap B_R(0) \cap F_\tau) \geq \vol\left(\conv\left(\left\{\hat{v}_1, \ldots, \hat{v}_{n+1}\right\}\right)\right) \geq \left(R + 2^{\poly(\enc{P})}\right)^{-n} \cdot (c')^n \cdot 2^{-\poly(\enc{P})}.
    \]
    Thus, by choosing
    \[
        r = \min \left\{2^{-\poly(\enc{P})}, \left(R + 2^{\poly(\enc{P})}\right)^{-n} \cdot \left(\frac{\varepsilon}{2L}\right)^n \cdot 2^{-\poly(\enc{P})}\right\}
    \]
    we get that
    \begin{equation}\label{EQ:taulargeimpliesvolumelarge}
        \tau > f_{\min} + \frac{\varepsilon}{2} \quad \Longrightarrow \quad \vol(P \cap B_R(0) \cap F_\tau) > r,
    \end{equation}
    since if $\tau > f_{\min} + \frac{\varepsilon}{2}$, then $c > \frac{\varepsilon}{2}$ and thus $\vol(P \cap B_R(0) \cap F_\tau) > r$.
    Note that, using that the Lipschitz constant $L$ satisfies $L = 2^{\poly(\enc{f},\, \log(R))}$, we get
    \[
        \log(1/r) = \poly(n,\, \enc{P},\, \log(R),\, \log(1/\varepsilon)).
    \]
    
    \begin{algorithm}
    \caption{}\label{ALG:algorithmfulldimensionalcase}
    \begin{algorithmic}[1]
    \Input A polynomial $f$, a matrix $A \in \Q^{m \times n}$, a vector $b \in \Q^m$, a bound $R$, an error parameter $\varepsilon$
    \Output A point $\tilde{x} \in P$ with $f(\tilde{x}) \leq \min_{x \in P \cap B_R(0)} f(x) + \varepsilon$.
    \State Let $\tau_\ell = -2^{\poly(\enc{f}, \, \log(R))}$ and $\tau_r = 2^{\poly(\enc{f}, \, \log(R))}$.
    \While{$\tau_r-\tau_\ell \geq \frac{\varepsilon}{2}$}
    \State $\tau = \frac{\tau_r + \tau_\ell}{2}$
    \State{\parbox[t]{0.964\linewidth}{Run the algorithm from \Cref{PROP:ellipsoidforconvexset} with input $r$, $R$ and a strong separation oracle for the set ${P \cap B_R(0) \cap F_\tau}$.}}
    \If{output is $x \in P \cap B_R(0) \cap F_\tau$ (i.e. we are in case (i))}
    \State $\tau_r = \frac{\tau_r + \tau_\ell}{2}$
    \Else{ (i.e. we are in case (ii))}
    \State $\tau_\ell = \frac{\tau_r + \tau_\ell}{2}$
    \EndIf
    \EndWhile
    \State{Run the algorithm from \Cref{PROP:ellipsoidforconvexset} with input $r$, $R$ and a strong separation oracle for the set $P \cap B_R(0) \cap F_{\tau_r}$.}
    \State{\Return the point $\Tilde{x}$ that is output by this algorithm}
    \end{algorithmic}
    \end{algorithm}

    \paragraph{Algorithm.}
    Using binary search, we want to use this to get a point $\Tilde{x} \in P$ with the guarantee that ${f(\Tilde{x}) \leq \min_{x \in P \cap B_R(0)} f(x) + \varepsilon}$.
    Consider \Cref{ALG:algorithmfulldimensionalcase}.
    This algorithm and the analysis are a standard way to move from using the ellipsoid method to solve a feasibility problem to solving an optimization problem.
    It can for example be found in \cite[Chapter 13]{Vishnoi:convexoptimization}.
    We include it here for completeness and in order to carry out the bit complexity arguments in detail, which is needed for our result.
    
    We claim that the algorithm satisfies the following two invariants:
    \begin{itemize}
        \item For $\tau_r$, we always have that the output of the algorithm from \Cref{PROP:ellipsoidforconvexset} belongs to case (i) (given $r$, $R$ and a separation oracle to $P \cap B_R(0) \cap F_{\tau_r}$ as input).
        \item For $\tau_\ell$, we always have $\tau_\ell \leq f^* + \frac{\varepsilon}{2}$.
    \end{itemize}
    This is true in the beginning because by Lipschitzness of $f$ and since $f(0) \leq 2^{\poly(\enc{f}, \, \log(R))}$, we know that on $B_R(0)$ the value of $f$ is in $[-2^{\poly(\enc{f}, \, \log(R))}, 2^{\poly(\enc{f}, \, \log(R))}]$.
    Thus, for the first choice of $\tau_r$, we have $F_{\tau_r} \supseteq B_R(0)$ and hence $\vol(P \cap B_R(0) \cap F_{\tau_r}) = \vol(P \cap B_R(0)) > r$ (using \eqref{EQ:taulargeimpliesvolumelarge} for $\tau = \infty$).
    For the first choice of $\tau_\ell$, we have $\tau_\ell \leq f^*$.
    Furthermore, this stays true during the whole algorithm because we only update $\tau_r$ in case (i) of the algorithm from \Cref{PROP:ellipsoidforconvexset}, i.e. exactly when the invariant stays true.
    If we update $\tau_\ell$, we are in case (ii) of the algorithm from \Cref{PROP:ellipsoidforconvexset}.
    Then we need to have $\vol(P \cap B_R(0) \cap \tau_\ell) < r$ (for the updated $\tau_\ell$) and hence $\tau_\ell \leq f_{\min} + \frac{\varepsilon}{2}$ by \eqref{EQ:taulargeimpliesvolumelarge}.

    Thus, the algorithm is well-defined (i.e., after the while-loop the algorithm from \Cref{PROP:ellipsoidforconvexset} does in fact output a point $\tilde{x}$). This point is in $P \cap B_R(0) \cap F_{\tau_r}$ for the final value of $\tau_r$. Hence, $f(\Tilde{x}) \leq f^* + \tau_r$.
    Furthermore, we have $\tau_r \leq \tau_\ell + \frac{\varepsilon}{2} \leq f^* + \varepsilon$ and thus we output a point $\Tilde{x} \in P \cap B_R(0)$ with
    \[
        f(\Tilde{x}) \leq f^* + \varepsilon = \min_{x \in P \cap B_R(0)} f(x) + \varepsilon.
    \]

    It remains to argue the runtime of the algorithm.
    We have $\poly(\enc{f},\, \log(R), \, \log(1/\varepsilon))$ iterations of the while-loop (we start with $\tau_r - \tau_\ell = 2^{\poly(\enc{f}, \, \log(R))}$, end with $\tau_r - \tau_\ell \leq \frac{\varepsilon}{2}$ and we half this distance in every step).
    This then also implies that ${\bc(\tau_r), \bc(\tau_\ell) \leq \poly(\enc{f},\, \log(R), \, \log(1/\varepsilon))}$.
    So, the separation oracle for $P \cap B_R(0) \cap F_{\frac{\tau_r + \tau_\ell}{2}}$ needs time $\poly(\enc{f},\, \enc{P},\, \log(R),\, \log(1/\varepsilon),\, k)$ on an input of bit length $k$.
    Thus, one execution of the algorithm from \Cref{PROP:ellipsoidforconvexset} takes time
    \[
        \poly(n,\, \log(R),\, \log(1/r),\, \enc{f},\, \enc{P},\, \log(1/\varepsilon)).
    \]
    Using $\log(1/r) = \poly(n,\, \enc{P},\, \log(R),\, \log(1/\varepsilon))$ and $n \leq \enc{f}$, this runtime is
    \[
        \poly(\enc{f},\, \enc{P},\, \log(R),\, \log(1/\varepsilon)).
    \]
    Thus, since we have $\poly(\enc{f},\, \log(R),\, \log(1/\varepsilon))$ iterations the overall runtime of the algorithm is also
    \[
        \poly(\enc{f},\, \enc{P},\, \log(R),\, \log(1/\varepsilon)),
    \]
    which completes the proof of \Cref{PROP:applicationellipsoid} for the full-dimensional case.
\end{proof}

Finally, we want to prove \Cref{PROP:applicationellipsoid} for a general (not necessarily full-dimensional) polyhedron by reducing this to the full-dimensional case.
\begin{proof}[Proof of \Cref{PROP:applicationellipsoid} (general case)]
    By \Cref{PROP:getaffinehullofpolyhedron}, we can compute vectors $v_1, \ldots, v_k$ such that $\aff(P) = \aff(\{v_1, \ldots, v_k\})$.
    Then, we also have that
    \[
        \aff(P) = v_1 \oplus \spann{v_2-v_1, \ldots, v_k-v_1}.
    \]
    By \Cref{LEM:LAGS}, we can replace the $v_2-v_1, \ldots, v_k - v_1$ by orthogonal vector $\hat{v}_2, \ldots, \hat{v}_k$ (this is not strictly necessary but makes the following argument simpler).
    Define the following affine map
    \[
        L: \R^{k-1} \to \aff(P) \subseteq\R^n, \quad x' \mapsto v_1 + \sum_{i=2}^k x'_{i-1} \hat{v}_i.
    \]
    This map is a bijection from $\R^{k-1}$ to $\aff(P)$. In fact, the inverse is
    \begin{equation}\label{EQ:ellispoidreductioninversemap}
        L^{-1}: \aff(P) \to \R^{k-1}, \quad x \mapsto (\langle x-v_1, \hat{v}_2\rangle/\|\hat{v}_2\|^2, \ldots, \langle x-v_1, \hat{v}_k\rangle/\|\hat{v}_k\|^2).
    \end{equation}
    We can write $L$ as $L(x') = v_1 + Bx'$ for an appropriate matrix $B \in \Q^{n \times (k-1)}$.
    We define
    \[
        P' = L^{-1}(P) = \{x' \in \R^{k-1} : v_1 + Bx' \in P\} = \{x' \in \R^{k-1}: (AB)x' \leq b - Av_1\} \subseteq \R^{k-1}
    \]
    and
    \[
        f': \R^{k-1} \to \R, \quad x' \mapsto f'(x') = f(L(x')).
    \]
    Note that $P'$ is a polyhedron in $\R^{k-1}$ and $f'$ is a $(k-1)$-variate polynomial.

    On the polyhedron $P$, the algorithm from \Cref{PROP:getaffinehullofpolyhedron} runs in time $\poly(\enc{P})$ (recall that by \Cref{LEM:seporacleforP} the separation oracle for $P$ runs in time $\poly(k, \, \enc{P})$ on an input of size~$k$).
    Also the orthogonalization then takes time $\poly(\enc{P})$ by \Cref{LEM:LAGS}.
    Thus, we can bound the bit length of $B$ as $\bc(B) \leq \poly(\enc{P})$ and thus also $\enc{P'} \leq \poly(\enc{P})$.
    Furthermore, this also implies that $\enc{f'} \leq \poly(\enc{f}, \, \enc{P})$.

    Since the $v_1, \ldots, v_k$ are affinely independent, $P'$ is full-dimensional ($L^{-1}(v_1), \ldots, L^{-1}(v_k)$ are affinely independent points in $\aff(P')$ and thus by a dimension argument, $\aff(P') = \R^{k-1}$).
    Let $R > 0$. Note that for $x \in P \cap B_R(0)$, using \eqref{EQ:ellispoidreductioninversemap}, we have
    \[
        \|L^{-1}(x)\| \leq \sum_{i=2}^{k} \frac{|\langle x-v_1 , \hat{v}_i\rangle|^2}{\|\hat{v_i}\|^2} \leq (R + \|v_1\|) \cdot \sum_{i=2}^k \frac{1}{\|\hat{v}_i\|}.
    \]
    Since $\bc(v_1), \bc(\hat{v}_i) \leq \poly(\enc{P})$, we have $\|v_1\|, \|\hat{v}_2\|^{-1}, \ldots, \|\hat{v}_k\|^{-1} \leq 2^{\poly(\enc{P})}$. Thus, by defining ${R' = (R+1) \cdot 2^{\poly(\enc{P})}}$ we get that $L^{-1}(P \cap B_R(0)) \subseteq P' \cap B_{R'}(0)$.

    Thus, by the proof of the full-dimension case (applied with $R'$, $\varepsilon$, $f'$ and $P'$), we can compute a point $\Tilde{x}' \in P'$ with $f(\Tilde{x}') \leq \min_{x' \in P' \cap B_{R'}(0)} f'(x') + \varepsilon$ in time
    \[
        \poly(\enc{f'},\, \enc{P'},\, \log(R'),\, \log(1/\varepsilon)) \leq \poly(\enc{f},\, \enc{P},\, \log(R),\, \log(1/\varepsilon)).
    \]
    Since $P \cap B_R(0) \subseteq L(P' \cap B_{R'}(0))$, we have $\min_{x' \in P' \cap B_{R'}(0)} f'(x') \leq \min_{x \in P \cap B_R(0)} f(x)$.
    Given a point~$\Tilde{x}'$ as above, we thus get a point $\Tilde{x} = L(\Tilde{x}') \in P$ with
    \[
        f(\Tilde{x}) = f'(\Tilde{x}') \leq \min_{x' \in P' \cap B_{R'}(0)} f'(x') + \varepsilon \leq \min_{x \in P \cap B_R(0)} f(x) + \varepsilon.\qedhere
    \]
\end{proof}

\bibliographystyle{amsalpha}
\bibliography{convpol}

\end{document}